\documentclass{l4dc2026}

\title[Scalable Koopman Reachability with Conformal Coverage Guarantees]{Scalable Data-Driven Reachability Analysis and Control via Koopman Operators with Conformal Coverage Guarantees}
\usepackage[utf8]{inputenc}    
\usepackage{amsmath,amssymb}
\usepackage{mathtools}
\usepackage{caption}
\usepackage{subcaption}
\usepackage{times}
\usepackage{hyperref}
\usepackage{makecell} 
\usepackage{algorithm2e}
\SetKwInOut{Input}{Input}
\SetKwInOut{Output}{Output}
\usepackage{booktabs}
\usepackage{xcolor} 
\newtheorem{problem}{Problem}
\usepackage{siunitx}
\usepackage{algpseudocode}
\usepackage{graphicx}
\usepackage{algorithm}
\algrenewcommand\algorithmicindent{0.8em}
\usepackage{wrapfig}
\usepackage{tikz}
\usepackage{contour}
\usepackage{ulem}

\contourlength{0.8pt}

\newcommand{\myuline}[1]{%
  \uline{\phantom{#1}}%
  \llap{\contour{white}{#1}}%
}

\author{%
 \Name{Devesh Nath}$^{*1}$ \Email{dnath7@gatech.edu}\\
 \Name{Haoran Yin}$^{*1}$ \Email{hyin95@gatech.edu}\\
 \Name{Glen Chou}$^2$ \Email{chou@gatech.edu}\\
 \addr Georgia Institute of Technology. $^1$ School of ECE, $^2$ Schools of Cybersecurity \& Privacy \& Aerospace Eng.%
}

\begin{document}

\maketitle

\algrenewcommand\algorithmicrequire{\textbf{Input:}}
\algrenewcommand\algorithmicensure{\textbf{Output:}}
\algrenewcommand\algorithmiccomment[1]{\textcolor{blue}{//~#1}}

\newcommand{\infKoop}{\mathcal{K}}
\newcommand{\ka}{K_{A}}
\newcommand{\kb}{K_{B}}
\newcommand{\enc}{\phi}
\newcommand{\dec}{\psi}
\newcommand{\tnfl}{\Gamma}
\newcommand{\B}{\mathcal{B}}
\newcommand{\dz}{$\delta z_t$}
\newcommand{\zr}{z^{ref}_t}
\newcommand{\ur}{u^{ref}_t}
\newcommand{\uf}{u^{ff}_t}
\newcommand{\zref}{z^{\text{ref}}}
\newcommand{\uref}{u^{\text{ref}}}
\newcommand{\delfz}{\delta z}
\newcommand{\delfu}{\delta u}
\newcommand{\lqrmat}{G}
\newcommand{\xrefdistr}{\text{MP}}
\newcommand{\ncscore}{R}

\newcommand{\ctrlSpace}{\mathcal{U}}
\newcommand{\stateSpace}{\mathcal{X}}
\newcommand{\ctrlSpaceDim}{\mathbb{R}^m}
\newcommand{\stateSpaceDim}{\mathbb{R}^n}

\begin{abstract}
\label{abstract}
\looseness-1We propose a scalable reachability-based framework for probabilistic, data-driven safety verification of unknown nonlinear dynamics. We use Koopman theory with a neural network (NN) lifting function to learn an approximate linear representation of the dynamics and design linear controllers in this space to enable closed-loop tracking of a reference trajectory distribution. Closed-loop reachable sets are efficiently computed in the lifted space and mapped back to the original state space via NN verification tools. To capture model mismatch between the Koopman dynamics and the true system, we apply conformal prediction to produce statistically-valid error bounds that inflate the reachable sets to ensure the true trajectories are contained with a user-specified probability. These bounds generalize across references, enabling reuse without recomputation. Results on high-dimensional MuJoCo tasks (11D Hopper, 28D Swimmer) and 12D quadcopters show improved reachable set coverage rate, computational efficiency, and conservativeness over existing methods.
\end{abstract}

\begin{keywords}
  safety-critical control, reachability, Koopman operators, conformal prediction
\end{keywords}

\begin{figure}[ht]
    \centering
    \includegraphics[width=\linewidth]{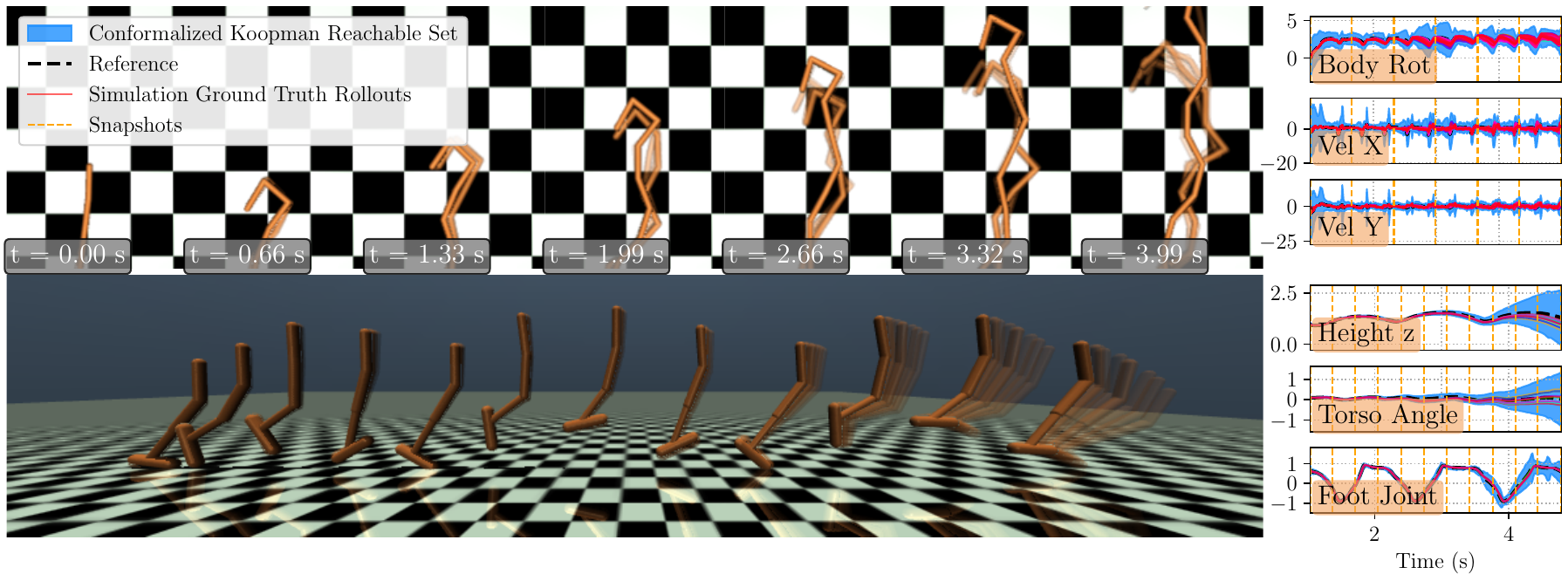}
    \caption{Executed trajectories / reachable sets for a 28D MuJoCo swimmer (a); 11D hopper (b).}
    \vspace{-12pt}
    \label{fig:hopper}
\end{figure}

\section{Introduction}
\label{introduction}

Computing reachable sets, i.e., the set of states reachable over a time horizon, is key to robot safety. However, reachability analysis for robots is difficult as they often (i) lack analytical models, (ii) have high-dimensional nonlinear dynamics, and (iii) require reasoning over long horizons. Even if analytical dynamics are available, challenges (ii) and (iii) cause nonlinear reachability tools \citep{bansal2017HJReach, DBLP:conf/cpsweek/Althoff15} to become prohibitive or to suffer from excessive overapproximation over long horizons \citep{Rober2024CARV}. Data-driven reachability \citep{10772635datadrivenreach, DBLP:journals/fmsd/BakBHKLP25} can be used without analytical models, but can be data-inefficient or lack guarantees that the estimated sets contain all possible trajectories. 
In contrast, linear reachability scales to long horizons and high dimensions \citep{10.1145/3302504.3311792}. The Koopman framework learns linear dynamics in a lifted state space, e.g., via neural networks (NNs), to approximate the nonlinear system \citep{koop_survey}. While computing reachable sets for the Koopman linearization can be more scalable than using the original nonlinear dynamics, existing methods only provide guarantees for the approximate lifted dynamics. Thus, \textit{safety may not hold for the true system} \citep{DBLP:conf/adhs/BakBDGP21, DBLP:journals/fmsd/BakBHKLP25}. Moreover, existing Koopman reachability methods analyze a \textit{single} dynamical system. However, robots often require multiple controllers to complete different tasks, which forces existing methods to recompute reachable sets from scratch for each new controller, which is inefficient. Reusing previously computed information would enable more efficient reachability analysis across controllers.

\looseness-1
To close these gaps in scalability, speed, and reusability, we propose \textbf{ScaRe-Kro} (\underline{Sca}lable \underline{Re}achability via \underline{K}oopman Ope\underline{r}at\underline{o}rs), which uses Koopman theory with an NN lifting function to perform reachability analysis and control. Unlike prior Koopman reachability methods for autonomous systems, we also design a linear controller in this lifted space to track a distribution of reference trajectories, each completing a different task. ScaRe-Kro efficiently computes closed-loop reachable sets under the induced \textit{linear tracking error dynamics} in the lifted Koopman space and maps them to the original state space via NN verification tools. To account for Koopman model error, we apply conformal prediction (CP) to derive statistically-valid error bounds. By inflating the mapped reachable sets with these CP bounds, they are guaranteed to contain the \textit{true system trajectories} with a user-specified probability (e.g., 97.5\%). Moreover, as the CP bounds are calibrated across the trajectory distribution, they can be reused across different references and thus closed-loop systems, unlike prior Koopman approaches that certify only a single autonomous system. Overall, using linear reachability and CP, our method enables scalable data-driven probabilistic verification and control, improving both efficiency and conservativeness over prior work. Our contributions are:

\begin{enumerate}
  \setlength{\itemsep}{0pt}
  \setlength{\parskip}{0pt}
  \setlength{\parsep}{0pt}
  \setlength{\topsep}{0pt}
  \setlength{\partopsep}{0pt}
  \item A scalable Koopman-based reachability framework for unknown nonlinear dynamics, using NN lifting functions and NN verification to estimate reachable sets in the original state space.
  \item An approach for inflating the resulting Koopman reachable sets using CP, to obtain probabilistic guarantees of overapproximation for the \textit{true} dynamics.
  \item Control design in the lifted space for closed-loop tracking of a reference trajectory distribution, with CP-based inflation bounds calibrated to be reusable across the distribution.
  \item Extensive evaluation, showing reliable verification for high-dimensional robotic systems (up to 28D), outperforming baselines in safety rate, computation speed, and conservativeness.
\end{enumerate}

\section{Related Work}
\label{related_work}

\paragraph{Reachability via Linear Relaxation}
\looseness-1
While exact reachable set computation for nonlinear systems is typically intractable, reachable set overapproximations (RSOAs) can often be efficiently computed to certify safety \citep{DBLP:journals/arcras/AlthoffFG21}. When the closed-loop system involves NN components (e.g., learned dynamics or controller), RSOAs can be computed via linear relaxation-based NN verification (NNV) tools (e.g., CROWN \citep{Zhang2018CROWN}) \citep{Zhang2018CROWN, Everett2021NFL, Jafarpour2024AkashNFL}. Here, conservativeness grows with the number of nonlinearities \citep{sage-mp} and can be excessive for nonlinear robot dynamics. Symbolic and one-shot methods mitigate this but raise runtime and memory costs \citep{Chen2023One-Shot, Rober2024CARV, Everett2021NFL}. We instead model the robot with Koopman-linearized dynamics, reducing the number of nonlinearities in the CG and improving scalability (i.e., computation time, conservativeness).

\paragraph{Koopman Reachability}
\looseness-1
Koopman \citep{koopman_theory} and Carleman \citep{doi:10.1137/1.9781611976847.1, Bayer2023} linearizations accelerate reachability but introduce approximation error. Moreover, while error bounds exist, these derivations are limited to quadratic or polynomial dynamics \citep{DBLP:journals/pnas/LiuKKLTC21, DBLP:journals/corr/abs-1711-02552, doi:10.1137/1.9781611976847.1}, which do not apply to many robot dynamics models. Prior Koopman reachability methods also lack reachable set coverage calibration \citep{DBLP:journals/access/ThapliyalH22} or guarantee containment only for a finite set of observed trajectories \citep{DBLP:conf/cdc/KochdumperB22}, failing to ensure safety for previously-unseen trajectories of the \textit{true} system. Prior methods also do not support NN liftings and consider autonomous systems, which corresponds to fixing a single controller. Changing controllers, as often done in robotics, requires repeating the calibration \citep{DBLP:conf/rp/ForetsS21, DBLP:conf/adhs/BakBDGP21, DBLP:journals/fmsd/BakBHKLP25}, which is inefficient. In contrast, we support NN liftings, provide probabilistic coverage guarantees for the \textit{true} dynamics, and integrate control design, enabling data-efficient reachability across multiple controllers.

\paragraph{Data-Driven Reachability}
\looseness-1
Data-driven reachability uses black-box trajectory data. Some methods adapt Hamilton-Jacobi analysis \citep{reachblackboxdynamical}, but this is costly for high-dimensional systems. Others learn reachability functions \citep{neureach} but require dense data to accurately generalize, or yield ellipsoidal RSOAs \citep{10772635datadrivenreach} assuming privileged system data (e.g., Lipschitz constants). Sampling-based methods \citep{pmlr-v168-lew22a} are limited to short horizons. Scenario optimization provides probabilistic guarantees \citep{pmlr-v242-dietrich24a} but is slower than Koopman-based methods \citep{DBLP:journals/fmsd/BakBHKLP25}.
Moreover, these methods assume a fixed controller, yet robots often switch controllers to achieve different goals, altering the closed-loop dynamics and forcing reachable sets to be recomputed. Trajectory-tracking reachability provides reusable tracking-error sets across reference trajectories \citep{DBLP:conf/cdc/HerbertCHBFT17, DBLP:conf/wafr/SinghCHTP18, DBLP:journals/ral/KnuthCOB21}. Data-driven variants learn contraction metrics \citep{DBLP:conf/cdc/ChouOB21, DBLP:conf/wafr/ChouOB22, 10161001} for tracking-based reachability of unknown systems, but finding such metrics is generally difficult. We build on this work, yielding Koopman reachable sets that are fast to compute, scalable to long horizons, and reusable for various closed-loop dynamics induced by tracking different trajectories.
\section{Preliminaries}
\label{preliminary}
We consider unknown, black-box nonlinear systems, either in open-loop \eqref{eq:open_loop_dynamics} or closed-loop \eqref{eq:closed_loop_dynamics}:

\vspace{-20pt}
\begin{subequations}
\begin{align}
    x_{t+1} &= f(x_t, u_t),
    \label{eq:open_loop_dynamics}
\end{align}

\vspace{-30pt}
\begin{align}
    x_{t+1} &= f(x_t, \pi(x_t, t)) \doteq \tilde f(x_t),
    \label{eq:closed_loop_dynamics}
\end{align}
\end{subequations}
\vspace{-20pt}

\noindent with timestep $t \in \{0, \dots, T\} \doteq \mathcal{T}$, state $x_t \in \stateSpace \subseteq \stateSpaceDim$, control $u_t \in \ctrlSpace \subseteq \ctrlSpaceDim$, and feedback control law $\pi: \mathcal{X} \times \mathcal{T} \rightarrow \mathcal{U}$. The initial state $x_0$ lies in set $\mathcal{X}_0 \subseteq \stateSpace$. Denote $\textsf{Int}(\underline{a}, \overline{a}) = \{a \mid \underline{a} \le a \le \overline{a}\}$, as an interval, where $\underline{a}, \overline{a} \in \mathbb{R}^A$ and inequalities hold element-wise, i.e., $\underline{a}_i \le a_i \le \overline{a}_i$, $1 \le i \le A$. Denote the $a$-ball as $\B_a(c) \doteq \{ x \mid \Vert x - c \Vert_\infty \le a\}$ and sequence $q_{1:Q} \doteq \{q_1, q_2, \ldots, q_Q\} = \{q_i\}_{i=1}^Q$. We use $\oplus$ and $\ominus$ to denote the Minkowski sum and difference, respectively. 

\looseness-1
\paragraph{Koopman Operators} As our work relies on Koopman operators, we discuss the basic concepts here. Consider autonomous discrete-time nonlinear dynamics $x_{t+1} = f(x_t)$. A Koopman operator $\infKoop$ is an infinite-dimensional linear operator that evolves the observables $\enc^\infty$ of the state $\infKoop \enc^\infty(x) = \enc^\infty \circ f(x)$ according to $\enc^\infty(x_{t+1}) = \infKoop \enc^\infty(x_t)$, where $\circ$ denotes function composition. As infinite-dimensional Koopman operators are intractable to obtain in practice, we define an approximation to $\infKoop$ using a finite-dimensional lifting function $\enc: \mathcal{X} \rightarrow \mathcal{Z} \subseteq \mathbb{R}^{l}$ and matrix $\ka \in \mathbb{R}^{l \times l}$, which define lifted dynamics $z_{t+1} = \ka z_{t}$,  where $z_t = \enc(x_t)$ is the lifted state and $z_{t} \in \mathcal{Z} \subseteq \mathbb{R}^{l}$. Here, $l$ can be arbitrarily selected by the choice of the lifting function $\enc$. For systems with control input \eqref{eq:open_loop_dynamics}, the lifted dynamics can be similarly written as $z_{t+1} = \ka z_{t} + \kb u_t$, where $\kb \in \mathbb{R}^{l \times m}$. In practice, $\enc$, $K_A$, and $K_B$ can be learned from data, with the lifting function $g$ parameterized using polynomials or NNs \citep{koop_survey}. In this work, we utilize an NN-based autoencoder architecture, where the encoder $\enc: \mathcal{X} \rightarrow \mathcal{Z}$ defines the lifting function and the decoder $\dec: \mathcal{Z} \rightarrow \hat{\mathcal{X}}$ defines its learned inverse, i.e., $\dec(\enc(x)) \approx x$; if $\dec$ is a perfect inverse, $\dec(\enc(x)) = x$ and $\hat{\mathcal{X}} = \mathcal{X}$. 

\paragraph{Reachability via NN Verification (NNV)} We use NNV to compute reachable sets for closed-loop systems involving NN-based lifting and inverse functions. NNV tools represent an NN as a computational graph (CG) -- a directed acyclic graph encoding the sequence of operations applied to an input \citep{Rober2024CARV}. For some CG $G$ with input set $\mathcal{S} \subseteq \mathbb{R}^{n_i}$ and output $G(\mathcal{S}) \subseteq \mathbb{R}^{n_o}$, we can compute a guaranteed overapproximation of its image $G(\mathcal{S})$ using CROWN-based \citep{Zhang2018CROWN} tools like \texttt{auto\_LiRPA} \citep{Xu2020AutoLiRPA}. These tools provide guaranteed affine lower and upper bounds, $\underline{G}$ and $\overline{G}$, on the output $G(\mathcal{S})$ for any interval $\mathcal{S}$, formalized in this proposition:

\begin{proposition}[CG Robustness \citep{Xu2020AutoLiRPA}]
\label{thm:cg_robustness}
\looseness-1For some CG $G$ and interval $\mathcal{S} \doteq \{s \in \mathbb{R}^{n_i} \mid \underline{s} \le s \le \overline{s}\}$, there are affine functions $\underline{G}$, $\overline{G}$ such that $\forall s \in \mathcal{S}$, $\underline{G}(s) \leq G(s) \leq \overline{G}(s)$. 
The inequalities hold element-wise and $\underline{G}(s) = \Psi s + \alpha$,  $\overline{G}(s) = \Phi s + \beta$, with $\Psi, \Phi \in \mathbb{R}^{n_o \times n_i}$ and $\alpha, \beta \in \mathbb{R}^{n_o}$.
\end{proposition}

\looseness-1
\noindent Given controller $\pi$, dynamics \eqref{eq:open_loop_dynamics}, and initial set $\mathcal{X}_0$, the $T$-timestep exact reachable set is denoted as $\mathcal{X}_{0:T} \doteq \{\mathcal{X}_0, f(\mathcal{X}_t, \pi(\mathcal{X}_t,t)) \}_{t=0}^{T-1}\}$. While exact reachability analysis is generally intractable for nonlinear systems \citep{DBLP:journals/arcras/AlthoffFG21}, it is often feasible to compute a reachable set overapproximation (RSOA) $\overline{\mathcal{X}}_{0:T}$, which satisfies $\mathcal{X}_t \subseteq \overline{\mathcal{X}}_t$ for all $t \in \{0,\ldots,T\}$. To obtain $\overline{\mathcal{X}}_{0:T}$, Prop. \ref{thm:cg_robustness} can be applied to the CG $F$ of the $T$-step composition ${\tilde f^{T}}(\cdot) \doteq \tilde{f} \circ \dotsb \circ \tilde{f}(\cdot)$ of the closed-loop dynamics $\tilde f$ \eqref{eq:closed_loop_dynamics}, i.e., $F(x_0) \doteq \{\tilde f(x_t) \doteq f(x_t, \pi(x_t, t)) \}_{t=0}^{T}$, with \texttt{auto\_LiRPA} computing the RSOA.

\paragraph{Conformal Prediction (CP)}
Our RSOA computations also rely on CP. In this work, we adopt split CP, which partitions i.i.d.\ input-output pairs $(v^{(i)}, y^{(i)}) \in \mathcal{V}\times\mathcal{Y}$ from a dataset $\mathcal{D} \doteq\{(v^{(i)}, y^{(i)})\}_{i=1}^{L+K}$ into a size-$L$ training set $\mathcal{D}_T$ and a size-$K$ calibration set $\mathcal{D}_C$. For a given prediction function $\mu:\mathcal{V}\!\to\!\mathcal{Y}$, we compute a nonconformity score $R^{(i)}\!\doteq s(y^{(i)},\mu(v^{(i)}))$ for each of the $K$ calibration pairs in $\mathcal{D}_C$, where $s: \mathcal{Y}\times \mathcal{Y} \rightarrow \mathbb{R}_{\ge 0}$ is a chosen nonconformity score function. As an example, we can express the $\ell_2$ prediction error of a learned dynamics model $\hat f$ in this framework by defining $v = (x_t, u_t)$, $y = x_{t+1}$, $\mu(v) = \hat f(v)$, and $s(y, \mu(v)) = \Vert y - \mu(v)\Vert_2$. Given a miscoverage level $\delta\in(0,1)$, a threshold $C$ is defined as the ($1-\delta$)-quantile of the calibration scores, $C \doteq \text{Quantile}_{1-\delta}(R^{(1)},\dots,R^{(K)},\infty)$. This threshold defines the prediction set $\mathcal{Y}(v) \doteq \{\,y\in\mathcal{Y} \mid s(y,\mu(v))\le C \,\}$, which satisfies the marginal coverage guarantee $\mathbb{P}(y^{(0)}\in\mathcal{Y}(v^{(0)})) \ge 1-\delta$ for an unseen data-point $(y^{(0)}, v^{(0)})$ that is exchangeable with $\mathcal{D}$ \citep{cp_survey}.

\section{Problem Statement}
\label{problem_statement}
\looseness-1We are given black-box access to $f$ \eqref{eq:open_loop_dynamics}, i.e., we do not know its analytical form and its output can only be estimated via data. We are a planner $P: \mathcal{X} \rightarrow \mathcal{X}^{T+1} \times \mathcal{U}^T$ that generates open-loop reference trajectories $(x^\textrm{ref}_{0:T}, u^\textrm{ref}_{0:T-1})$ that satisfy \eqref{eq:open_loop_dynamics}, starting from an initial condition $x_0^\textrm{ref} \in \mathcal{X}_0 \ominus \mathcal{B}_\epsilon(0)$ that is uniformly distributed over $\mathcal{X}_0 \ominus \mathcal{B}_\epsilon(0)$.  We assume nothing further about $P$; it can be a learned or traditional planner. At runtime, the true initial state $x_0$ may be perturbed from $x_0^\textrm{ref}$, i.e., $x_0 \in \mathcal{B}_\epsilon(x_0^\textrm{ref})$, where $x_0$ is uniformly distributed over $\mathcal{B}_\epsilon(x_0^\textrm{ref})$. To be robust to this error, we design a controller $\pi: \mathcal{X}\times \mathcal{T}\rightarrow \mathcal{U}$ that tracks references from $P$ and compute high-probability RSOAs $\overline{\mathcal{X}}_{0:T}$ for the resulting closed-loop dynamics \eqref{eq:closed_loop_dynamics}, such that $\mathbb{P}\big(\bigwedge_{t=0}^T (x_t \in \overline{\mathcal{X}}_t) \big) \ge 1-\delta$. This process should scale to high-dimensional nonlinear systems over long horizons (e.g., $n \ge 25$, $T \ge 200$). We solve:
\begin{problem}[Model learning \& controller design]\label{prob:koopman}
Given $N$ length-$(T+1)$ \textbf{open-loop}, dynamically-feasible trajectories generated by $P$ via black-box queries of \eqref{eq:open_loop_dynamics}, forming dataset $\mathcal{D}=\{(x_{0:T}^{(i)},u_{0:T-1}^{(i)}\}_{i=1}^N$, (a) train a Koopman model $(\enc,\dec,\ka,\kb)$ and (b) use the Koopman model to design a time-varying trajectory-tracking controller $\pi: \mathcal{X} \times \mathcal{T} \rightarrow \mathcal{U}$ to stabilize to references generated by $P$.
\end{problem}
\begin{problem}[($1-\delta$)-confident RSOA computation]\label{prob:verification}
Given $M$ length-$(T+1)$ \textbf{closed-loop} trajectories generated by tracking references from $P$ with the controller $\pi$ from Prob. \ref{prob:koopman}, compute a ($1-\delta$)-confident RSOA $\overline{\mathcal{X}}_{0:T}$ for the closed-loop dynamics \eqref{eq:closed_loop_dynamics}, which guarantees for any random reference $(x^\textrm{ref}_{0:T}, u^\textrm{ref}_{0:T-1})$ produced by the planner $P$ with $x_0^\textrm{ref} \in \mathcal{X}_0$ and initial state $x_0 \in \mathcal{B}_\epsilon(x_0^\textrm{ref})$, the closed-loop trajectory remains in $\overline{\mathcal{X}}_{0:T}$ with probability at least $1-\delta$, i.e., $\mathbb{P}\big(\bigwedge_{t=0}^T (x_t \in \overline{\mathcal{X}}_t) \big) \ge 1-\delta$.
\end{problem}

\section{Methodology}
\label{methodology}

Our method (Fig. \ref{fig:blk_diag}, Alg. \ref{alg:online_ckrs}) learns a Koopman operator (Sec. \ref{sec:koop_train}), uses it for control design (Sec. \ref{sec:controller_design}), performs fast linear reachability analysis on the closed-loop lifted Koopman dynamics (Sec. \ref{sec:reachability}), and inflates these reachable sets via CP to obtain probabilistic coverage guarantees (Sec. \ref{sec:conformalize_KRS}).

\subsection{Koopman Operator Training}
\label{sec:koop_train}
To solve Prob. \ref{prob:koopman}, we train a Koopman operator model $(\enc, \dec, \ka, \kb)$ on a dataset $\mathcal{D} \doteq \{(x_{0:T}^{(i)}, u_{0:T-1}^{(i)})\}_{i=1}^N$ of open-loop trajectories from a planner $P$ (details in Sec. \ref{sec:experiments}). We use a composite loss $\mathcal{L} = \sum_{i=1}^N(\lambda_1 \mathcal{L}_1^{(i)} + \lambda_2 \mathcal{L}_2^{(i)})$. The first term, $\mathcal{L}_{\text{1}}^{(i)} = \sum_{j=0}^{T} \big[\|x_{j}^{(i)} - \dec (\enc (x_{j}^{(i)}))\|^2 + \|\enc(x_j^{(i)}) - \enc(\dec(\enc(x_j^{(i)})))\|^2\big]$, enforces autoencoder accuracy via state reconstruction error and a latent consistency loss. The second term, $\mathcal{L}_{\text{2}}^{(i)} = \sum_{j=t+1}^{t+H} \big[\| x_{j}^{(i)} - \dec (\check z_{j}^{(i)})\|^2 + \|\enc (x_{j}^{(i)}) - \check z_j^{(i)}\|^2\big]$, where $\check z_j^{(i)} \doteq \ka \check z_{j-1}^{(i)} + \kb u_{j-1}^{(i)}$ and $\check z_0^{(i)}=\enc(x_0^{(i)})$, is a multi-step (horizon $H$) dynamics loss penalizing prediction errors from the linear model ($\ka$, $\kb$) in both latent and state space.The latent consistency term is crucial for long-horizon prediction accuracy. We initialize $\ka$ as identity and $\kb$ with Xavier uniform, train via learning rate annealing and weight decay, and set ($\enc, \dec$) to have ReLU or GELU activations.
\subsection{Koopman Tracking LQR Design}
\label{sec:controller_design}
To reliably imitate state/control reference trajectories generated by the planner $P$ \textit{in the lifted state space} $\mathcal{Z}$, a tracking controller $\pi$ is required to stabilize against initial condition perturbations $x_0 - x_0^\textrm{ref}$ and to reject disturbances due to model error in the learned lifted dynamics $z_{t+1} = \ka z_t + \kb u_t$. To obtain $\pi$, we use the linearity of the lifted dynamics to design a linear quadratic regulator (LQR) that tracks trajectories in $\mathcal{Z}$. Given a reference state-space trajectory $x_{0:T}^{\text{ref}}$ produced by $P$, we map it to the lifted state space as $z_{0:T}^{\text{ref}} \doteq \{\enc(x_t^{\text{ref}})\}_{t=0}^T$. The corresponding feedforward control trajectory $u_{0:T-1}^{\text{ref}}$ which minimizes the lifted-state imitation error is then computed as:
\begin{equation}
\label{eq:u_ff}
u_{t}^{\text{ref}} = \kb^{\dagger} (z_{t+1}^{\text{ref}} - \ka z_{t}^{\text{ref}}),
\end{equation}
where $\kb^{\dagger}$ is the pseudo-inverse of $\kb$. The error states $\delfz_{t} = z_{t} - \zref_{t}$ and controls $\delfu_{t} = u_{t} - \uref_{t}$ follow linear error dynamics $\delfz_{t+1} = \ka \delfz_{t} + \kb \delfu_{t}$, starting from an initial error $\delfz_0 = \enc(x_0) - \zref_0$, where $x_0 \in \B_{\epsilon}(x_{0}^{\text{ref}})$. 
We solve an LQR problem, finding the feedback controls $\delfu_{0:T-1}$ that minimize the finite-horizon quadratic cost $\delfz_T^\top Q_T \delfz_T + \sum_{t=0}^{T-1} (\delfz_t^\top Q \delfz_t + \delfu_t^\top R \delfu_t)$, subject to the linear dynamics $\delfz_{t+1} = \ka \delfz_t + \kb \delfu_t$, where $Q \succ 0$, $R \succ 0$, and $Q_T \succ 0$. In particular, we solve a backward Riccati recursion to obtain the optimal feedback gains $\lqrmat_{0:T-1}$. This yields the tracking control law in the lifted state space \eqref{eq:u_latent_cl}, which can be applied on the nonlinear system \eqref{eq:open_loop_dynamics} as \eqref{eq:u_real_cl} by measuring and encoding the state $x_t$:

\begin{minipage}{0.48\textwidth}
\begin{equation}
\label{eq:u_latent_cl}
u_t = \uref_t - \lqrmat_t (z_t - \zref_t),
\end{equation}
\end{minipage}
\hfill
\begin{minipage}{0.48\textwidth}
\begin{equation}
\label{eq:u_real_cl}
u_t = \uref_t - \lqrmat_t (\enc(x_t) - \zref_t).
\end{equation}
\end{minipage}
\subsection{Reachability Analysis for Closed-Loop Koopman Dynamics}\label{sec:reachability}
\looseness-1
The high computational cost of nonlinear reachability motivates our use of Koopman operators. Standard RSOA computation repeatedly overapproximates nonlinearities at each timestep in both $f$ and $\pi$ (which typically must be nonlinear to stabilize $f$), leading to high conservativeness. A $T$-step RSOA requires $2T$ nonlinear function evaluations, causing reachable sets to become vacuously loose as $T$ increases. To reduce conservativeness, partition-based methods split the input set $\overline{\mathcal{X}}_t$, propagate each partition element, and union the result; however, this incurs computational and memory costs that increases with the cardinality of the partition. In contrast, Koopman-based reachability improves scalability and reduces conservativeness without partitioning overhead by reducing the number of nonlinearities. Mapping to the lifted space $\mathcal{Z}$ creates linear dynamics, eliminating the need to bound nonlinearities during propagation. The linear dynamics also permit a linear tracking controller. Consequently, computing an RSOA $\overline{\mathcal{X}}_t$ for any timestep $t$ requires bounding only a single application of the nonlinear encoder ($\enc$) and decoder ($\dec$), reducing the number of nonlinear bounding operations from $2T$ to just $2$. 
To apply this insight formally for RSOA computation, we define the Tracking Koopman Feedback Loop (TKFL), denoted $\tnfl$, as a CG the describes the closed-loop dynamics for tracking $z_{0:T}^{\text{ref}} \doteq \{\enc(x_t^{\text{ref}})\}_{t=0}^T$ in the lifted state space using the feedback control law \eqref{eq:u_real_cl}.  Specifically, the TKFL $\tnfl : \mathcal{X}_0 \times \mathcal{U}^T \times \mathcal{Z}^{T+1} \rightarrow \hat{\mathcal{X}}^{T+1}$ maps the initial state $x_0 \in \B_{\epsilon}(x_{0}^{\text{ref}})$, the feedforward control trajectory $u_{0:T-1}^{\text{ref}}$ \eqref{eq:u_ff}, and the lifted reference $z_{0:T}^{\text{ref}}$ to the decoded trajectory $\hat{x}_{0:T} \doteq \{\dec(z_{t})\}_{t=0}^T$ representing the predicted closed-loop trajectory \eqref{eq:latent_rollout_tnfl} in the original state space, where the lifted states $z_t$ are updated via the closed-loop lifted dynamics \eqref{eq:latent_closed_loop}: 

\begin{equation}
\label{eq:latent_rollout_tnfl}
    \hspace{-23pt}\hat{x}_{0:T} = \{x_0, \{\dec (\tilde f_z(z_t))\}_{t=0}^{T-1}\},
\end{equation}
\begin{equation}
\label{eq:latent_closed_loop}
    \hspace{-30pt} z_{t+1} = \ka z_t + \kb (u_t^\textrm{ref} - G_t(z_t - z_t^\textrm{ref})) \doteq \tilde f_z(z_t).\hspace{-20pt}
\end{equation}

\begin{figure}[ht]
  \centering
  \includegraphics[width=0.85\textwidth]{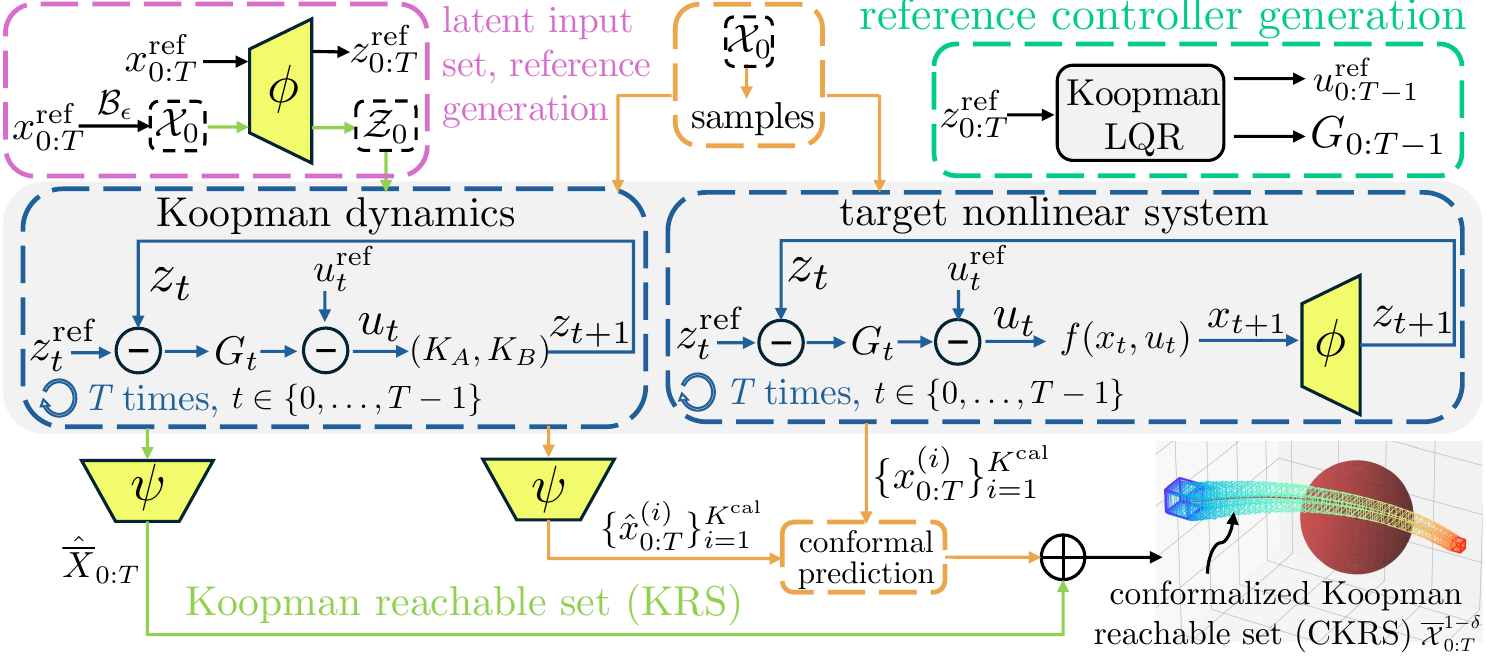}
  \caption{A block diagram that describes the flow of our method (ScaRe-Kro). Green and orange arrows represent KRS and CP computations, respectively. }
  \label{fig:blk_diag}
\end{figure}

To summarize, the CG is constructed as follows (see Fig. \ref{fig:blk_diag} for visualization). First, we map the initial state via $\enc$ into the latent space, computing $z_0 = \enc(x_0)$. Second, we unroll the $T$-fold recursive application of the single-step latent dynamics, $z_{t+1} = \ka z_t + \kb u_t$,  where $u_t$ is the closed-loop tracking control computed using \eqref{eq:u_latent_cl}, to obtain the latent trajectory $z_{0:T} = \{z_0, z_2, \cdots, z_T\}$. Third, $z_{0:T}$ is mapped back to the original state space via the decoder $\dec$ to obtain $\hat{x}_{0:T}$. 
The CG $\tnfl$ is passed to the \texttt{auto\_LiRPA} library \citep{Xu2020AutoLiRPA}, which enables set-based bound propagation through $\tnfl$ using CROWN \citep{Zhang2018CROWN}. \texttt{auto\_LiRPA} computes vector-valued affine bounding functions, $\underline{\tnfl}$ and $\overline{\tnfl}$, which enables us to compute an RSOA $\hat{\overline{\mathcal{X}}}_{0:T}$ for the decoded closed-loop lifted dynamics \eqref{eq:latent_rollout_tnfl} by considering specific components of $\tnfl$, $\underline{\tnfl}_t$ and $\overline{\tnfl}_t$, that bound the states reached under \eqref{eq:latent_rollout_tnfl} at timestep $t$, starting from any initial condition $x_0 \in {\mathcal{X}}_0$. This resulting hyper-rectangular set $\hat{\overline{\mathcal{X}}}_{0:T}$, which we refer to as the Koopman reachable set (KRS), is found by solving the optimization \eqref{eq:Koopman_RSOA}, which is computed efficiently by \texttt{auto\_LiRPA}:
\begin{equation}\small
    \hspace{-5pt}
    \hat{\overline{\mathcal{X}}}_{0:T} \doteq \big\{ \textsf{Int}\big( \textstyle\min_{\substack{x_0 \in \mathcal{X}_0}} \underline{\tnfl}_t(x_0, \uref_{0:T-1}, \zref_{0:T}), \max_{\substack{x_0 \in \mathcal{X}_0}} \overline{\tnfl}_t(x_0, \uref_{0:T-1}, \zref_{0:T}) \big) \big\}_{t=0}^{T}
    \hspace{-3pt}
    \label{eq:Koopman_RSOA}
\end{equation}
\looseness-1The KRS $\hat{\overline{\mathcal{X}}}_{0:T}$ overapproximates $\hat{x}_{0:T}$ so that $\hat{x}_{t} \in \hat{\overline{\mathcal{X}}}_{t}$, for all $t \in \mathcal{T}$, ensuring the decoded closed-loop lifted dynamics \eqref{eq:latent_rollout_tnfl} remain in $\hat{\overline{\mathcal{X}}}_{0:T}$. Formally, we state the following result (proof in App. \ref{app:proofs}):

\setcounter{theorem}{0}
\begin{lemma}[KRS overapproximation]
\label{lem:koopman_containment}
\renewcommand{\thelemma}{1}
Let $\tnfl$ denote the closed-loop CG of TKFL, and let $\mathcal{X}_0$ be the initial set. Suppose the KRS $\hat{\overline{\mathcal{X}}}_{0:T}$ is defined as in \eqref{eq:Koopman_RSOA}. Then, for any initial state $x_0 \in \mathcal{X}_0$, the decoded closed-loop trajectory \eqref{eq:latent_rollout_tnfl}, i.e.,
$\hat{x}_{0:T} = \tnfl(x_0, \zref_{0:T}, \uref_{0:T-1})$ satisfies $\hat{x}_t \in \hat{\mathcal{X}}_t$ for all $t \in \mathcal{T}$.
\end{lemma}

\looseness-1\noindent Crucially, the KRS overapproximates the reachable set for the decoded Koopman linearized dynamics \eqref{eq:latent_rollout_tnfl}, \textit{not the true nonlinear dynamics} \eqref{eq:closed_loop_dynamics}, due to Koopman model error. This motivates Sec. \ref{sec:conformalize_KRS}.

\subsection{Conformalizing Koopman Reachable Sets}
\label{sec:conformalize_KRS}
\looseness-1
We now explain how to inflate the KRS to guarantee it contains the true system \eqref{eq:closed_loop_dynamics} trajectories with a user-specified probability $1-\delta$. We use CP to compute a high-probability upper bound on the state-space prediction error of the decoded closed-loop lifted dynamics \eqref{eq:latent_rollout_tnfl}. Specifically, we adapt the single nonconformity score approach (SNSA) \citep{cp_survey} to compute state- and dimension-dependent error bounds using a calibration dataset $\mathcal{D}_C = \{(x_{0:T}^{(i)}, \hat x_{0:T}^{(i)})\}_{i=1}^{K^\textrm{cal}}$ of trajectories of the true closed-loop dynamics $x_{0:T}$ \eqref{eq:closed_loop_dynamics} and the decoded closed-loop lifted dynamics $\hat x_{0:T}$  \eqref{eq:latent_rollout_tnfl}. To construct each data tuple in $\mathcal{D}_C$, we sample a reference trajectory $x_{0:T}^{\text{ref},(i)}$ from planner $P$ and an initial state $x_0^{(i)}$ uniformly from $\B_{\epsilon}(x_{0}^{\text{ref}})$. We then compute the latent reference trajectory $z_{0:T}^{\text{ref}, (i)} = \{\enc(x_{t}^{\text{ref},(i)})\}_{t=0}^T$ and the feedforward control sequence $u_{0:T-1}^{\text{ref}, (i)}$ using \eqref{eq:u_ff}. We use $u_{0:T-1}^{\text{ref}, (i)}$ to calculate the trajectory returned by 1) tracking $z_{0:T}^{\text{ref}}$ in the lifted state space (\ref{eq:latent_rollout}) and decoding to the state space via $\dec$ and 2) executing the Koopman controller on the true \textit{nonlinear dynamics} (\ref{eq:actual_rollout}):

\noindent\begin{minipage}[b]{0.54\textwidth}
\begin{equation}\small
\label{eq:latent_rollout}
\hspace{-5pt}\hat{x}_{1:T}^{(i)} =
\bigl\{\, \dec\bigl(\ka z_{t} + \kb\bigl(\uref_t - G_t(z_t - \zref_t)\bigr)\bigr) \,\bigr\}_{t=0}^{T-1}\hspace{-6pt}
\end{equation}
\end{minipage}%
\hspace{\fill} 
\begin{minipage}[b]{0.44\textwidth} 
\begin{equation}\small
\label{eq:actual_rollout}
\hspace{-1pt}x_{1:T}^{(i)} = \{f(x_t, \uref_t - G_t (\enc(x_t) - \zref_t)\}_{t=0}^{T-1}
\end{equation}
\end{minipage}

\looseness-1
The prediction error for each state dimension $j \in \{1, \dots, n\}$ at time $t$ is denoted $e_{t,j}^{(i)} = x_{t,j}^{(i)} - \hat{x}_{t,j}^{(i)}$. This gives a calibration set of error trajectories $\mathcal{D}_{E} = \{\{e_{t}^{(i)}\}_{t=0}^{T}\}_{i=1}^{K^{\text{cal}}}$, where each $e_t^{(i)} \in \mathbb{R}^n$. Each error trajectory $e^{(i)}$ is normalized into a single non-conformity score $\ncscore^{(i)} = \max_{t,j} ( \lambda_{t,j} |e_{t,j}^{(i)}|)$, which is the maximum weighted error over all time steps $t \in \{0, \dots, T\}$ and state dimensions $j \in \{1, \dots, n\}$. Here, the normalization constant $\lambda_{t,j} \doteq 1 / (e_{t,j}^{\text{max}} + \sigma) \in \mathbb{R}_{> 0}$ is computed using a separate normalization dataset $\mathcal{D}_N$. This dataset, $\mathcal{D}_N$, consists of $M^{\lambda}$ additional error trajectories, $\{e^{(i)}\}_{i=K^{\text{cal}}+1}^{K^{\text{cal}}+M^{\lambda}}$, collected using the same process as the calibration data. Here, $\sigma$ is a small positive constant to prevent division by zero. The term $e_{t,j}^{\text{max}}$ is the dimension-wise and time-wise maximum absolute error observed across $\mathcal{D}_N$, i.e., $e_{t,j}^{\text{max}} = \max_{i} |e_{t,j}^{(i)}|$ where $i \in \{K^{\text{cal}}+1, \dots, K^{\text{cal}}+M\}$. 

\looseness-1
Now, given a user-defined miscoverage rate $\delta \in (0, 1)$, we compute the $(1-\delta)$ quantile of the non-conformity scores $C = \text{Quantile}_{1-\delta}(\{\ncscore^{(i)}\}_{i=1}^{K^{\text{cal}}} \cup \{\infty\})$ by finding the $p$-th smallest value of the scores $\{\ncscore^{(i)}\}_{i=1}^{K^{\text{cal}}}$, where $p \equiv \lceil (K^{\text{cal}} + 1)(1-\delta) \rceil$. The final, time- and dimension-dependent error bound $\bar e_{t,j} \in \mathbb{R}_{\ge 0}$ is obtained by reversing the normalization $\bar e_{t,j} = C / \lambda_{t,j}$. These bounds ensure that for an unseen test data-point $(x_{0:T}^{(0)}, \hat x_{0:T}^{(0)})$, the true prediction error $|e_{t,j}^{(0)}|$ will be contained within the bound $\bar e_{t,j}$ for all $t$ and $j$, with probability at least $1-\delta$, i.e.,
$$\mathbb{P}\big( |e_{t,j}^{(0)}| \le \bar e_{t,j}, \forall t \in \{0, \dots, T\}, \forall j \in \{1, \dots, n\} \big) \ge 1 - \delta$$

\noindent We use these bounds to inflate the KRS computed in \eqref{eq:Koopman_RSOA} to compute RSOAs for the original closed-loop dynamics \eqref{eq:closed_loop_dynamics} that are valid with probability at least $1-\delta$, using \eqref{eq:conformal_inflation} (proof in App. \ref{app:proofs}):

\begingroup
\setlength{\abovedisplayskip}{6pt}
\setlength{\belowdisplayskip}{6pt}
\setcounter{theorem}{0}
\begin{theorem}[CKRS Coverage Guarantee]
\label{thm:ckrs_coverage}
Let $\hat{\overline{\mathcal{X}}}_{0:T}$ be the KRS defined in \eqref{eq:Koopman_RSOA} and let $\bar e_t = [\bar e_{t,1}, \ldots, \bar e_{t,n}]^\top \in \mathbb{R}^n$. Define the conformalized Koopman reachable set (CKRS) as

\begin{equation}
\overline{\mathcal{X}}_{0:T}^{1-\delta} \equiv \bigl\{ \hat{\overline{\mathcal{X}}}_t \oplus \mathrm{diag}(\bar e_t) \mathcal{B}_1(0) \bigr\}_{t=0}^T,
\label{eq:conformal_inflation}
\end{equation}

\noindent where $\mathrm{diag}(\bar e_t) \in \mathbb{R}^{n\times n}$ is a diagonal matrix with $\bar e_t$ on its diagonal. Then, for a new reference trajectory $x_{0:T}^{\textrm{ref}, (0)}$ drawn from the same distribution, the CKRS contains the true closed-loop system trajectory generated by \eqref{eq:closed_loop_dynamics} with probability at least $1-\delta$, i.e., $\mathbb{P}\big( \bigwedge_{t=0}^T x_t^{(0)} \in \overline{\mathcal{X}}_t^{1-\delta} \big) \ge 1-\delta$.
\end{theorem}
\endgroup

\paragraph{Offline Calibration}
As discussed so far, the method in Sec. \ref{sec:conformalize_KRS} is reference-specific, as the calibration trajectories $\mathcal{D}_C$ and error bounds $\bar e_t$ hold specifically for a single reference $x_{0:T}^{\text{ref}}$. We propose an offline approach that pre-computes a single, global set of error bounds $\bar e_{t,j}$ valid for any reference $x_{0:T}^{\text{ref}}$ sampled from the planner $P$. This is done by building a larger, more diverse calibration set. We first sample $N^{\text{ref}}$ reference trajectories $\{x_{0:T}^{\text{ref},(i)}\}_{i=1}^{N^{\text{ref}}} \sim P$. For each reference, we generate a calibration trajectory by sampling $x_0^{(i)} \sim \mathcal{B}_\epsilon(x_{0}^{ \text{ref},(i)})$, yielding an error trajectory $e_t^{(i)} = x_t^{(i)} - \hat x_{t}^{(i)}$, where $x_t^{(i)}$ and $\hat x_{t}^{(i)}$ are as defined in \eqref{eq:latent_rollout} and \eqref{eq:actual_rollout}, respectively. The SNSA procedure (computing $\ncscore^{(i)}$, $C$, and $\bar e_{t,j}$ via $\lambda_{t,j}$) is then performed exactly once on this entire aggregated set.

At runtime, given a new reference $x_{0:T}^{\text{ref}, (0)}$ sampled from $P$, two efficient computations are performed: 1) we compute the nominal KRS $\hat{\overline{\mathcal{X}}}_{0:T}$ specific to $x_{0:T}^{\text{ref}, (0)}$ via \eqref{eq:Koopman_RSOA}), and 2) we inflate this set using the pre-computed offline bounds $\bar e_t$ via the Minkowski sum $\overline{\mathcal{X}}_{t}^{1-\delta} = \hat{\overline{\mathcal{X}}}_{t} \oplus \textrm{diag}(\bar e_{t})\mathcal{B}_1(0)$, as defined in \eqref{eq:conformal_inflation}. This enables fast reachable set computation at runtime but may increase conservativeness of the resulting CKRS, since the bounds $\bar e_t$ must now hold uniformly for all references drawn from $P$, rather than for a single fixed reference. The validity of this offline approach relies on the exchangeability of the non-conformity scores, which is met since each score $\ncscore^{(i)}$ is a deterministic function of the data pair $(x_{0:T}^{\text{ref},(i)}, x_0^{(i)})$. Each of these pairs is sampled i.i.d. from the joint distribution defined by $x_{0:T}^{\text{ref}} \sim P$ and $x_0 \sim \mathcal{B}_\epsilon(x_{0}^{\text{ref}})$. Hence, the resulting scores are i.i.d. and thus exchangeable, ensuring that the standard CP guarantee $\mathbb{P}\big( \bigwedge_{t=0}^T x_t^{(0)} \in \overline{\mathcal{X}}_t^{1-\delta} \big) \ge 1-\delta$ holds.

\section{Experiments}
\label{sec:experiments}
\looseness-1We first compare our method against reachability baselines and then evaluate it on reachability analysis for a unicycle (3D), planar quadcopter (6D), 3D quadcopter (12D), MuJoCo \citep{todorov2012mujoco} hopper (11D), and MuJoCo swimmer (28D). All experiments are run on a desktop with an Intel i9-14900K CPU, 64\,GB RAM, and an NVIDIA RTX 4090 GPU. For the MuJoCo examples, we set the planner $P$ to be an expert PPO \citep{schulman2017proximalpolicyoptimizationalgorithms} policy; for the other systems, we use a trajectory optimizer on the analytical dynamics to serve as $P$. We report four key metrics: (i) \textit{reachability runtime} (s), where lower is better; (ii) \textit{average log-volume of reachable sets}, where lower indicates less conservativeness; (iii) \textit{safety rate over rollouts}, where higher is better; and (iv) \textit{coverage variability for CP-based verification}, where lower Beta-posterior variance of empirical coverage across calibration sizes indicates more stable coverage. Some metrics apply only to specific baselines (e.g., safety filters report safety rate only). Comprehensive results are presented in Table~\ref{tab:result_tab}.

\paragraph{Baselines}
\looseness-1We select baselines to highlight four gaps in data-driven reachability. (i) Learned reachability estimators \citep{deepreach} lack overapproximation guarantees, motivating formal NN verification. (ii) Existing hybrid CP-NNV (CP-NNV) \citep{cp_baseline} methods struggle to scale to high-dimensional, nonlinear dynamics and long horizons, motivating the Koopman lift. (iii) Pure CP-based reachability approaches \citep{cp_survey} suffer coverage degradation as calibration data or sample budgets shrink, motivating sample-efficient reachability analysis that remains reliable under limited data. (iv) Pure NNV-based RSOA computation \citep{sage-mp} provides deterministic guarantees but can be slow, motivating fast Koopman-based linear reachability.

\noindent\textit{\myuline{Unsafe behavior from learned reachability estimators.} } Under unicycle dynamics \eqref{eq:unicycle}, a value-function-based safety filter \citep{deepreach} -- an approximate reachability method -- fails to consistently enforce reach-avoid safety, with only 70\% of 250 rollouts satisfying the constraint (Fig. \ref{fig:deepreach}). Because the learned value function is an approximate neural PDE solution, small errors in the value and its gradient can cause safety violations, and the method lacks formal guarantees. In contrast, our method provides calibrated probabilistic guarantees that the closed-loop system remains in the CKRS computed via \eqref{eq:conformal_inflation} (100\% in our experiments) while preserving goal-reaching.

\noindent \textit{\myuline{Pure CP based verification.} }
Because our learned Koopman lift closely matches the training dynamics, the resulting KRS are already tight; conformalization then serves only to absorb the small residual model error. We compare to a pure SNSA baseline \citep{cp_survey}, which aggregates all heterogeneous trajectory and state errors into a single nonconformity score. As the horizon or state dimension grows, this forces the baseline to either inflate its sets or require much larger calibration sets to achieve the target $\delta$, causing it to suffer from slow convergence to the target coverage level. At practical data budgets (e.g., 1000 samples), a large coverage variance remains, e.g., as much as 5\% under-coverage of the true system trajectories (Fig. \ref{fig:pure_cp}), which can compromise system safety. In contrast, our method attains $100\%$ empirical coverage across all tested calibration set sizes, avoiding this variance-data trade-off and achieving target coverage with far fewer samples.

\begin{figure}[ht]
    \centering
    \includegraphics[width=\linewidth]{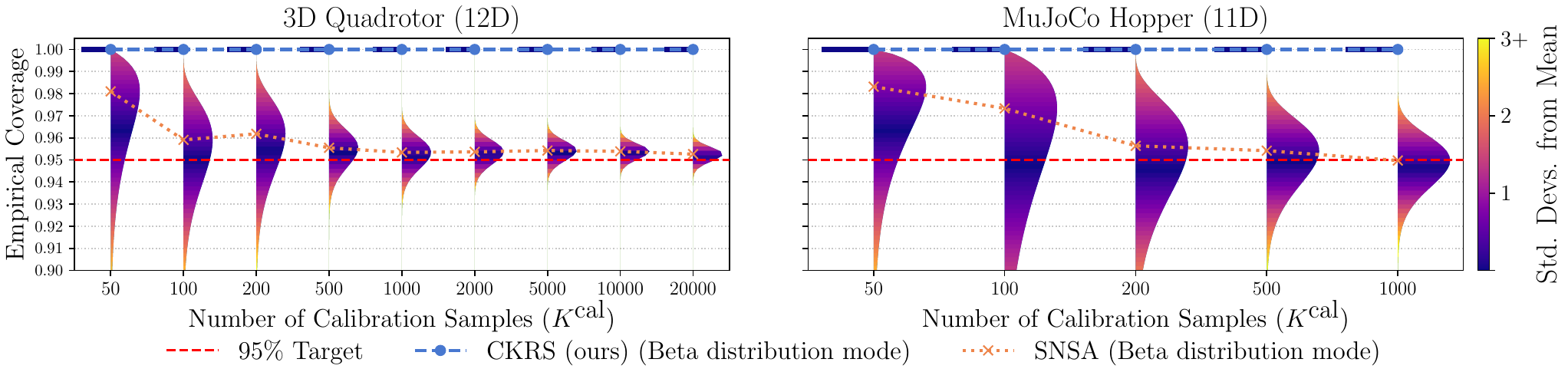}
    \caption{\looseness-1To visualize empirical reachable set coverage, we compute Beta posteriors for the SNSA baseline and ScaRe-Kro on the 3D quadcopter (left) and Hopper (right). The dotted lines indicate the mode of the Beta posterior. For SNSA, coverage rate slowly converges to the target $1-\delta$ as calibration dataset size increases, whereas ScaRe-Kro maintains $100\%$ coverage even for small $K^\textrm{cal}$.}
    \label{fig:pure_cp}
\end{figure}

\noindent \textit{\myuline{Hybrid Reachability methods (CP+NNV).} }
\looseness-1For a 3D quadcopter \eqref{eq:3D_Quad_dyn}, we compare with a hybrid data-driven/NNV reachability method \citep{cp_baseline} (denoted CP-NNV), that computes an RSOA by applying NNV on a learned NN trajectory predictor and using CP to bound RSOA underapproximation error induced by learning error. As shown in Fig.~\ref{fig:cpnnv}, CP-NNV yields much larger reachable sets (avg. log volume: -23.35), especially in attitude and velocity. This reflects difficulties in scaling NNV and CP to high-dimensional state spaces and long horizons, which we address with Koopman operators. Our method produces tighter RSOAs (avg. log volume: -40.66) while maintaining coverage, improving scalability with horizon and reducing conservativeness.

\noindent \textit{\myuline{Pure NNV Reachability.} }
On a unicycle model, we compare with our prior work \citep{sage-mp}, which used NNV tools for one-shot RSOA computation, requiring $51.95$s and achieving an average log-volume of $-3.98$. In contrast, our method computes the RSOA in $0.32$s with a log-volume of $-11.54$ (Table \ref{tab:result_tab}). This highlights that our method is much faster and also reduces conservativeness.

\begingroup
\setlength{\textfloatsep}{4pt}
\setlength{\intextsep}{4pt}
\setlength{\floatsep}{4pt}
\setlength{\abovecaptionskip}{2pt}
\setlength{\belowcaptionskip}{0pt}
\renewcommand{\arraystretch}{0.5}   
\setlength{\tabcolsep}{3pt}         
\setlength{\aboverulesep}{0pt}
\setlength{\belowrulesep}{0pt}
\setlength{\cmidrulesep}{2pt}

\newcommand{\thdr}[1]{\smash[b]{\shortstack{#1}}}

\begin{table}[!t]
\centering
\scriptsize
\caption{Numerical RSOA results for SCaRe-Kro on unicycle, quadcopters and MuJoCo models.}
\label{tab:result_tab}
\begin{tabular}{lcccccccccc}
\toprule
\textbf{Experiment} &
\textbf{\thdr{System\\dimension}} &
\textbf{\thdr{Time\\steps}} &
\textbf{\thdr{Confidence\\Level\\$(1-\delta)$}} &
\textbf{\thdr{CP\\time}} &
\textbf{\thdr{KRS\\time}} &
\textbf{\thdr{Total\\time\\(online)}} &
\textbf{\thdr{Avg. Log\\volume\\(online)}} &
\textbf{\thdr{Avg. Log\\volume\\(offline)}} &
\textbf{\thdr{Calibration\\dataset\\size}} &
\textbf{\thdr{Empirical\\coverage}} \\
\midrule
Unicycle & 3  & $100$ & $99.0\%$ & $0.038$s & $0.277$s & $0.315$s & $-11.543$ & $-8.690$ & $100$  &  $100\%$\\
2D Quad  & 6  & $100$ & $99.0\%$ & $0.055$s & $0.349$s & $0.404$s & $-21.491$ & $-5.333$  & $100$ & $100\%$ \\
3D Quad  & 12 & $200$ & $99.0\%$ & $0.175$s & $0.500$s & $0.675$s & $-42.954$ & $-35.275$ & $100$ & $100\%$ \\
Hopper   & 11 & $225$ & $99.0\%$ & $10.597$s & $0.387$s & $10.983$s & $-4.969$ & $-1.806$ & $100$ & $100\%$ \\
Swimmer  & 28 & $400$ & $95.0\%$ & $20.818$s & $0.593$s & $21.411$s & $32.595$ & $42.617$ & $100$ & $99.5\%$ \\
\bottomrule
\end{tabular}
\end{table}

\endgroup

\paragraph{Unicycle}
\looseness-1
We evaluate on CKRS computations on a unicycle model \eqref{eq:unicycle}, where $\enc$ and $\dec$ have 3 hidden layers with 128 neurons each, and $z \in \mathbb{R}^{10}$. As shown in Fig.~\ref{fig:car}, the certified reachable sets verify both safety and goal attainment. Full metrics are in Table~\ref{tab:result_tab}. We view this benchmark as a sanity check on a simple system, paving the way for more substantive scalability results.

\paragraph{Planar Quadcopter}
\looseness-1We evaluate a 6-state planar quadcopter \eqref{eq:2dquad} navigating around a circular obstacle for $T=100$. Here, $\enc$ and $\dec$ have 3 hidden layers with 128 neurons each, and $z \in \mathbb{R}^{24}$. We visualize offline-calibrated CKRSs for multiple reference trajectories in Fig.~\ref{fig:2d_quad}; metrics are in Table~\ref{tab:result_tab}. An NNV-only baseline \citep{sage-mp} required $59.85\,\mathrm{s}$ and achieved an average log-volume of $-18.709$. Overall, these results support our claim of reduced RSOA conservativeness. 

\paragraph{3D Quadcopter}
\looseness-1
\label{sec:3dquad}
We evaluate an analytical 12-state quadcopter \eqref{eq:3D_Quad_dyn} \citep{sabatino2015quadrotor} navigating around a $1\,\mathrm{m}$ radius obstacle for $T=200$. Here, $\enc$ and $\dec$ have 3 hidden layers with 128 neurons each, and $z \in \mathbb{R}^{24}$. The computed CKRS (Fig.~\ref{fig:3d_quad}) verifies safety and goal reaching; metrics are in Table~\ref{tab:result_tab}. To show that the offline CP inflation generalizes across controllers, we compute the CKRS for two references (Fig. \ref{fig:3d_quad}, one orange, one blue) using the same offline-calibrated bounds; we visualize additional CKRSs in Fig. \ref{fig:3d_quad_alt}. An NNV-only baseline \citep{sage-mp} is much slower and more conservative, requiring $2153.45\,\mathrm{s}$ with an average log-volume of $-30.32$. Overall, these results show our method's efficiency, tightness, and offline CP generalization for high-dimensional systems.

\begin{figure}[ht]
    \centering
    \includegraphics[width=\linewidth]{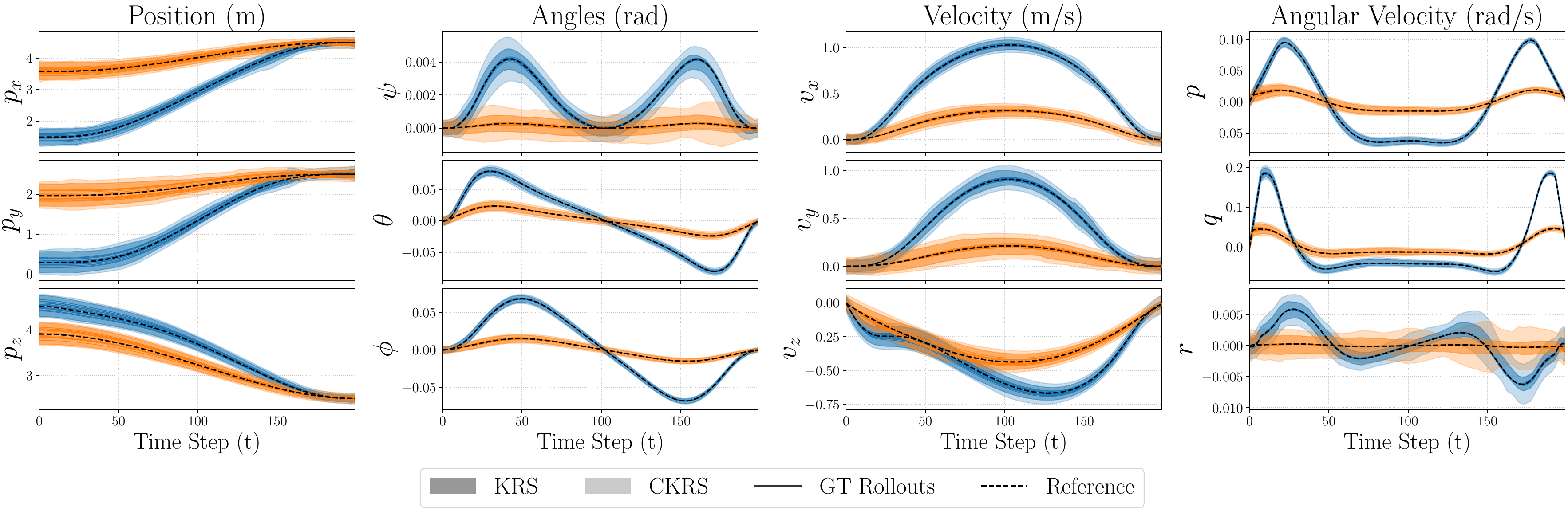}
    \caption{CKRS and KRS (pre-CP) computed for 3D quadcopter, plotted for each dimension.}
    \label{fig:3d_quad}
\end{figure}

\paragraph{MuJoCo: Hopper}
\looseness-1
We evaluate on MuJoCo’s Hopper-v5 \citep{todorov2012mujoco}, a nonlinear, underactuated system with black-box dynamics. A PPO agent (Stable-Baselines3, default hyperparameters) trained for $2\times 10^6$ steps is used as the planner $P$ that generates trajectories for the training dataset. Here, $\enc$ and $\dec$ have 3 hidden layers with 128 neurons each, and $z \in \mathbb{R}^{24}$. Numerical results are in Table~\ref{tab:result_tab} and RSOA slices for a few state dimensions are shown in Fig.~\ref{fig:hopper} (see Fig. \ref{fig:hopper_exp} for all 11 dimensions). This experiment highlights that our approach scales to black-box, high-dimensional contact-rich dynamics while maintaining low runtimes and probabilistic coverage guarantees.

\paragraph{MuJoCo: Customized Swimmer}
\looseness-1
We further assess scalability on a customized MuJoCo Swimmer-v5 \citep{todorov2012mujoco} with extended links, yielding a 28-dimensional state. The planner $P$ is a PPO agent (Stable-Baselines3 MLP) trained for $1\times10^{7}$ steps with modified hyperparameters (LR $1\times10^{-3}$, $\gamma=0.9999$, batch size $512$; others default). Here, $\enc$ and $\dec$ have 3 hidden layers with 128 neurons each, and $z \in \mathbb{R}^{38}$. Numerical results are in Table~\ref{tab:result_tab} and RSOA slices for a few state dimensions are shown in Fig.~\ref{fig:hopper} (see Fig. \ref{fig:swimmer} for all 28 dimensions). This experiment indicates that our approach maintains certified probabilistic safety while tractably computed RSOAs for high-dimensional systems over long horizons for black-box dynamical systems in MuJoCo.

\vspace{-16pt}
\section{Conclusion}
\label{conclusion}

We present ScaRe-Kro, a method for scalable reachability analysis of complex, nonlinear black-box dynamics that provides formal coverage guarantees. Our approach learns a lifted linear representation of the nonlinear dynamics through data, which facilitates tracking control design and the efficient propagation of closed-loop reachable sets. We then ``conformalize" this reachable set by inflating it such that rollouts of the target nonlinear system are guaranteed to stay within the inflated set with a specified probability (i.e., at least $1-\delta$). Across several benchmarks, including a unicycle, planar/3D quadrotors, MuJoCo hopper, and a 13-link MuJoCo swimmer, our method computes reachable sets efficiently and over long time horizons. Future work aims to: 1) remove the dependence on reference trajectories and 2) replace LQR with a robust control policy that satisfies constraints despite model error to increase the time horizon over which the system can be verified.

\bibliography{references}
\newpage
\section{Appendix}
\label{Appendix}

We provide proofs for our theoretical results (App. \ref{app:proofs}), an algorithm block detailing our method (App. \ref{app:algorithm}), and additional experiments and experimental details (App. \ref{app:parameters}).
\subsection{Proofs}\label{app:proofs}
\setcounter{theorem}{0}
\begin{lemma}[KRS overapproximation]
\label{lem:koopman_containment}
Let $\tnfl$ denote the closed-loop CG of TKFL, and let $\mathcal{X}_0$ be the initial set. Suppose the KRS $\hat{\overline{\mathcal{X}}}_{0:T}$ is defined as in \eqref{eq:Koopman_RSOA}. Then, for any initial state $x_0 \in \mathcal{X}_0$, the decoded closed-loop trajectory \eqref{eq:latent_rollout_tnfl}, i.e., 
$\hat{x}_{0:T} = \tnfl(x_0, \zref_{0:T}, \uref_{0:T-1})$ satisfies $\hat{x}_t \in \hat{\mathcal{X}}_t$ for all $t \in \mathcal{T}$.
\end{lemma}

\begin{proof}
The closed-loop graph $\tnfl$ is a computational graph $G$ as defined in Proposition~\ref{thm:cg_robustness}. Using the \texttt{auto\_LiRPA} library~\cite{Xu2020AutoLiRPA}, we compute sound, vector-valued affine bounding functions $\underline{\tnfl}_t$ and $\overline{\tnfl}_t$ (corresponding to $\underline{G}$ and $\overline{G}$ in Proposition~\ref{thm:cg_robustness}) for the $t$-th state output.

By Proposition~\ref{thm:cg_robustness}, these bounds are guaranteed to be sound. That is, for any specific input $x_0 \in \mathcal{X}_0$, the predicted state 
$\hat{x}_t = \tnfl_t(x_0, \zref_{0:T}, \uref_{0:T-1})$ satisfies
\[
\underline{\tnfl}_t(x_0, \zref_{0:T}, \uref_{0:T-1}) 
\le \hat{x}_t 
\le \overline{\tnfl}_t(x_0, \zref_{0:T}, \uref_{0:T-1}).
\]

The KRS $\hat{\overline{\mathcal{X}}}_t$ is defined in \eqref{eq:Koopman_RSOA} as the interval $\textsf{Int}(\hat{x}_t^L, \hat{x}_t^U)$, where $\hat{x}_t^L = \min_{x_0' \in \mathcal{X}_0} \underline{\tnfl}_t(x_0', \dots)$ and $\hat{x}_t^U = \max_{x_0' \in \mathcal{X}_0} \overline{\tnfl}_t(x_0', \dots)$.
By the definitions of the minimum and maximum, we have
\begin{align*}
    \hat{x}_t 
    &\ge \underline{\tnfl}_t(x_0, \dots) 
     \ge \min_{x_0' \in \mathcal{X}_0} \underline{\tnfl}_t(x_0', \dots) 
     = \hat{x}_t^L,\\
    \hat{x}_t 
    &\le \overline{\tnfl}_t(x_0, \dots) 
     \le \max_{x_0' \in \mathcal{X}_0} \overline{\tnfl}_t(x_0', \dots) 
     = \hat{x}_t^U.
\end{align*}
Hence, $\hat{x}_t \in \textsf{Int}(\hat{x}_t^L, \hat{x}_t^U) = \hat{\overline{\mathcal{X}}}_t$. Since this holds for all $t \in \mathcal{T}$, the trajectory $\hat{x}_{0:T}$ is contained within $\hat{\overline{\mathcal{X}}}_{0:T}$.
\end{proof}

\setcounter{theorem}{0}
\begin{theorem}[CKRS Coverage Guarantee]
\label{thm:ckrs_coverage}
Let $\hat{\overline{\mathcal{X}}}_{0:T}$ be the KRS defined in \eqref{eq:Koopman_RSOA} and let $\bar e_t = [\bar e_{t,1}, \ldots, \bar e_{t,n}]^\top \in \mathbb{R}^n$. Define the conformalized Koopman reachable set (CKRS) as
\begin{equation}
\overline{\mathcal{X}}_{0:T}^{1-\delta} \equiv \bigl\{ \hat{\overline{\mathcal{X}}}_t \oplus \mathrm{diag}(\bar e_t) \mathcal{B}_1(0) \bigr\}_{t=0}^T,
\label{eq:conformal_inflation}
\end{equation}
where $\mathrm{diag}(\bar e_t) \in \mathbb{R}^{n\times n}$ is a diagonal matrix with $\bar e_t$ on its diagonal. Then, for a new reference trajectory $x_{0:T}^{\textrm{ref}, (0)}$ drawn from the same distribution, the CKRS contains true closed-loop system trajectory generated by \eqref{eq:closed_loop_dynamics} with probability at least $1-\delta$:
\[
\mathbb{P}\Big( \bigwedge_{t=0}^T x_t^{(0)} \in \overline{\mathcal{X}}_t^{1-\delta} \Big) \ge 1-\delta.
\]
\end{theorem}

\begin{proof}
The proof combines a deterministic containment guarantee with a probabilistic coverage guarantee.

\noindent\textit{Step 1: KRS Containment.}  
By Lemma~\ref{lem:koopman_containment}, the decoded closed-loop lifted trajectory $\hat{x}_{t,j}^{(0)}$ is deterministically contained within the hyper-rectangular KRS:
\[
\hat{x}_{t,j}^{(0)} \in \hat{\overline{\mathcal{X}}}_{t,j} = \textsf{Int}(\hat{x}_{t,j}^L, \hat{x}_{t,j}^U), \quad \forall t \in \mathcal{T},\quad \forall j \in \{1,\ldots,n\}.
\]

\noindent\textit{Step 2: Probabilistic Coverage.}  
The CP guarantee requires exchangeability of the nonconformity scores $\{\ncscore^{(i)}\}_{i=0}^{K^{\mathrm{cal}}}$. In our framework, all calibration and test trajectories share the same reference $(z_{0:T}^{\mathrm{ref}}, u_{0:T-1}^{\mathrm{ref}})$ and controller $\lqrmat_t$, and their initial states $x_0^{(i)}$ are sampled i.i.d. from $\mathcal{B}_\epsilon(x_0^{\mathrm{ref}})$. Since the decoded closed-loop lifted trajectory \eqref{eq:latent_rollout}, true closed-loop trajectory \eqref{eq:actual_rollout}, and the nonconformity score are deterministic functions of $x_0^{(i)}$, the resulting scores $\ncscore^{(i)}$ are also i.i.d., and hence exchangeable. By the CP guarantee, the event
\[
\mathcal{E} \equiv \{ |e_{t,j}^{(0)}| \le \bar e_{t,j}, \forall t,j \}
\]
occurs with probability at least $1-\delta$.

\noindent\textit{Step 3: Combined Guarantee.}  
The true state can be decomposed as $x_{t,j}^{(0)} = \hat{x}_{t,j}^{(0)} + e_{t,j}^{(0)}$. When $\mathcal{E}$ occurs, we have
\[
\hat{x}_{t,j}^L - \bar e_{t,j} \le \hat{x}_{t,j}^{(0)} + e_{t,j}^{(0)} \le \hat{x}_{t,j}^U + \bar e_{t,j},
\]
which by definition implies $x_{t,j}^{(0)} \in \overline{\mathcal{X}}_{t,j}^{1-\delta}$. Since this holds for all $t$ and $j$, the trajectory $x_{0:T}^{(0)}$ is contained in the CKRS $\overline{\mathcal{X}}_{0:T}^{1-\delta}$ with probability at least $1-\delta$.
\end{proof}

\subsection{Algorithm}\label{app:algorithm}

Algorithm \ref{alg:online_ckrs} details the procedure for computing the CKRS. The process begins by lifting the planner-provided reference trajectory $x_{0:T}^\textrm{ref}$ into the latent space to compute the feedforward control $u^\textrm{ref}$ via least-squares, and synthesizing a time-varying LQR controller to track this reference. To account for the discrepancy between the learned Koopman dynamics and the true system, the algorithm generates calibration data from system rollouts and utilizes SNSA-based conformal prediction to derive statistically-valid error bounds $\bar e_t$. In parallel, the nominal KRS, $\overline{\mathcal{X}}_{0:T}$, is computed by constructing a computational graph of the closed-loop tracking dynamics in the lifted Koopman state space and computing reachable sets for the linear Koopman dynamics via \texttt{auto\_LiRPA}. Finally, the algorithm returns a conformalized Koopman reachable set (CKRS), $\mathcal{X}_{0:T}^{1-\delta}$, obtained by inflating the nominal KRS with the derived conformal error bounds via a Minkowski sum. Crucially, the CKRS guarantees coverage of the trajectories generated by the \textit{true system} with probability $1-\delta$.

\begin{algorithm}[]

\caption{Online Conformalized Koopman Reachable Set (CKRS)}
\label{alg:online_ckrs}
\begin{algorithmic}[1]
 \Require Koopman model ($\enc, \dec, \ka, \kb$); true dynamics $f$ (simulator); reference trajectory $x^{\text{ref}}_{0:T}$; initial set $\mathcal{X}_0 = \mathcal{B}_\epsilon(x_0^{\text{ref}})$; RSOA confidence $1-\delta$; calibration sizes $M^{\lambda}, K^{\text{cal}}$; cost matrices $Q, R$; verification tool $\texttt{auto\_LiRPA}$
 \Ensure Conformalized Koopman reachable set $\overline{\mathcal{X}}^{1-\delta}_{0:T}$ (CKRS)
 \Statex \Comment{1. Initialization}
 \State $z^{\text{ref}}_{0:T} \gets \enc(x^{\text{ref}}_{0:T})$ \Comment{encode reference state trajectory}
 \State $u^{\text{ref}}_{0:T-1} \gets$ Equation \eqref{eq:u_ff} \Comment{compute feedforward control sequence}
 \State $\lqrmat_{0:T-1} \gets$ \textrm{LQR}$(K_A, K_B)$ \Comment{compute LQR gains}
 \State $\pi(\cdot) \gets$ Equation \eqref{eq:u_latent_cl} \Comment{compute lifted controller using LQR and feedforward inputs}
 \Statex \Comment{2. Error bounds $\bar e_t$ via CP}
 \State $\mathcal{D}_E, \mathcal{D}_N \gets \texttt{GenerateCalibrationData}(M^{\lambda}, K^{\text{cal}})$
 \State $\bar e_t \gets \texttt{SNSA}(\mathcal{D}_N, \mathcal{D}_E)$
 \Statex \Comment{3. Nominal KRS $\hat{\overline{\mathcal{X}}}_{0:T}$}
 \State $\tnfl(\cdot) \gets \dec(\texttt{PropagateLatent}(\enc(\cdot), z^{\text{ref}}, u^{\text{ref}}, \lqrmat))$ \Comment{define CG}
 \State $\hat{\overline{\mathcal{X}}}_{0:T} \gets \texttt{AutoLiRPA}(\tnfl, \mathcal{X}_0)$ \Comment{compute KRS via CG}
 \State $\overline{\mathcal{X}}^{1-\delta}_{0:T} \gets \hat{\overline{\mathcal{X}}}_{0:T} \oplus \textrm{diag}(\bar e_t) \mathcal{B}_1(0)$ \Comment{4. combine to obtain CKRS, and return}
 \State \textbf{return} $\overline{\mathcal{X}}^{1-\delta}_{0:T}$
\end{algorithmic}
\end{algorithm}

\subsection{Experiment Details and Additional Experiments}\label{app:parameters}

The following section contains the the analytical dynamics we used in this paper (App. \ref{app:dynamics}) and additional experiment results. In particular, we benchmark our approach with a baseline Koopman dynamics learning approach that leverages polynomial liftings (App. \ref{app:polynomial}), a baseline learned reachable set prediction method (App. \ref{app:deepreach}), and a baseline conformal prediction-based reachable set computation method that does not leverage Koopman-based propagation (App. \ref{app:nnv}). We also present additional experiment figures, including full reachable set visualizations on a per-state breakdown for the unicycle system (App. \ref{app:unicycle}), planar quadcopter system (App. \ref{app:quadcopter}), MuJoCo hopper system (App. \ref{app:hopper}), and MuJoCo swimmer system (App. \ref{app:swimmer}).

\subsubsection{Dynamics. }\label{app:dynamics}
For unicycle experiments ($x \in \mathbb{R}^3$, $u \in \mathbb{R}^2$), we used the dynamics model, time-discretized via forward Euler ($dt =0.1$s)
\begin{equation}\label{eq:unicycle}
    \dot{\mathbf{x}}\;=\; f(\mathbf{x},\mathbf{u})=
    \begin{bmatrix}
        u_1 \cos(x_2) \\ u_1 \sin(x_2) \\ u_2
    \end{bmatrix}.
\end{equation}

\noindent For planar quadcopter experiments ($x \in \mathbb{R}^6$, $u \in \mathbb{R}^2$), we used the parameters $m=0.5\textrm{kg}$, $g=-9.81\textrm{m/s}^2$, $I_y=0.01\textrm{kg} \cdot \textrm{m}^2$ for the dynamics model, and time-discretized with $dt = 0.05\,\mathrm{s}$:
\begin{equation}\label{eq:2dquad}
    \dot{\mathbf{x}}\;=\; f(\mathbf{x},\mathbf{u})=\begin{bmatrix}
        x_4 \\ x_5 \\ x_6 \\ -\frac{u_1}{m} \sin(x_3)\\ g+\frac{u_1}{m}\cos(x_3)\\\frac{u_2}{I_y}
    \end{bmatrix}
\end{equation}

\noindent For 3D quadcopter experiments, we used the parameters $m=1$kg, $g=-9.81$ m/s$^2$, $I_{x,y,z}=[0.5,0.1,0.3]$ kg·m$^2$ for the dynamics, with a time-discretization of $dt = 0.025\,\mathrm{s}$:
\begin{equation}
    \dot{\mathbf{x}} \;=\; f(\mathbf{x},\mathbf{u})=
    \begin{bmatrix}
        \dot{x} \\
        \dot{y} \\
        \dot{z} \\
        q \cdot \sin(\phi)/\cos{\theta}+r\cdot \cos{\phi}/\cos{\theta} \\
        q \cdot \cos{\phi} - r\cdot \sin{\phi} \\
        p + q \cdot \sin{\phi}\cdot \tan{\theta} +r\cdot \cos{\phi} \cdot \tan{\theta} \\
        \frac{u_1}{m} \cdot (\sin{\phi} \cdot \sin{\psi} + \cos{\phi} \cdot \cos{\psi} \cdot \sin{\theta}) \\ 
        \frac{u_1}{m} \cdot (\cos{\phi} \cdot \sin{\phi} - \cos{\phi} \cdot \sin{\psi} \cdot \sin{\theta}) \\ 
        g + u_1 \cdot (\cos{\phi} \cdot \cos{\theta}) / m\\
        ((I_y - I_z) / I_x) \cdot q\cdot r + \frac{u_2}{I_x} \\
        ((I_z - I_x) / I_y) \cdot p\cdot r + \frac{u_3}{I_y} \\
        ((I_x - I_y) / I_z) \cdot p\cdot q + \frac{u_4}{I_z} \\
    \end{bmatrix}
    \label{eq:3D_Quad_dyn}
\end{equation}

\subsubsection{Comparison between neural Koopman liftings and polynomial liftings}\label{app:polynomial}
We compare root mean square error (RMSE) between a polynomial-lifted extended dynamic mode decomposition baseline and our neural lifting approach on the Hopper system. For the polynomial lift, we use PyKoopman \citep{pykoopman} with degree-3 monomials and identical training data/preprocessing. During open-loop evaluation, the polynomial model becomes numerically unstable, that is, its predictions grow without bound and diverge by the eighth timestep, so we report RMSE only over the first 5 prediction steps in Table \ref{tab:hopper_11d_rmse}. In contrast, our neural lifting yields a stable latent dynamics model that tracks accurately across three complete hops, enabling longer-horizon evaluation. These results highlight that for certain systems, e.g., those with strongly nonlinear, contact-rich dynamics (like the MuJoCo hopper), fixed polynomial bases are insufficient, whereas a neural network lifting function captures the necessary structure and remains stable over extended rollouts. Together, these results motivate the development of a Koopman reachable-set computation framework that incorporates neural networks in the loop, which we enable using \texttt{auto\_LiRPA}.

\begin{table}[ht]
\centering
\begin{tabular}{lrr}
\hline
\textbf{State} & \textbf{Polynomial Lifting} & \textbf{Ours} \\
\hline
z-coord (height)   & 2728.481591        & 0.009848 \\
torso angle        & 117.374908         & 0.004201 \\
thigh joint        & 2832.025934        & 0.013581 \\
leg joint          & 10791.162550       & 0.015690 \\
foot joint         & 8063.167273        & 0.017901 \\
vel x-coord        & 321013.541866      & 0.047526 \\
vel z-coord        & 1691803.060196     & 0.059763 \\
vel torso          & 4512148.680118     & 0.082879 \\
vel thigh          & 596074.452102      & 0.073777 \\
vel leg            & 6164730.041985     & 0.164258 \\
vel foot           & 135554800.971296   & 0.259034 \\
\hline
\end{tabular}
\caption{RMSE comparison for the hopper (11D) system on rollouts performed by the learned model and the polynomial fitted model.}
\label{tab:hopper_11d_rmse}
\end{table}

\subsubsection{Baseline comparison with \cite{deepreach}.}\label{app:deepreach}

Figure~\ref{fig:deepreach} offers a qualitative comparison of safety violations between the baseline DeepReach method \citep{deepreach} and ScaRe-Kro. While the main text statistics indicate a lower safety rate for the baseline, this visualization explicitly captures the ``reach-avoid" failure mode: the orange trajectories generated by the DeepReach value function clearly drift into the red obstacle region. This highlights the risks of relying on approximate neural PDE solutions for safety-critical constraints. In contrast, the reachable set (shown in blue) computed by our method creates a tight, verified corridor that contains the ground truth rollouts with high, calibrated probability, visually confirming that the controller successfully steers the dynamics clear of the hazard while adhering to the probabilistic bounds.

\begin{figure}[ht]
  \centering
  \includegraphics[width=1\linewidth]{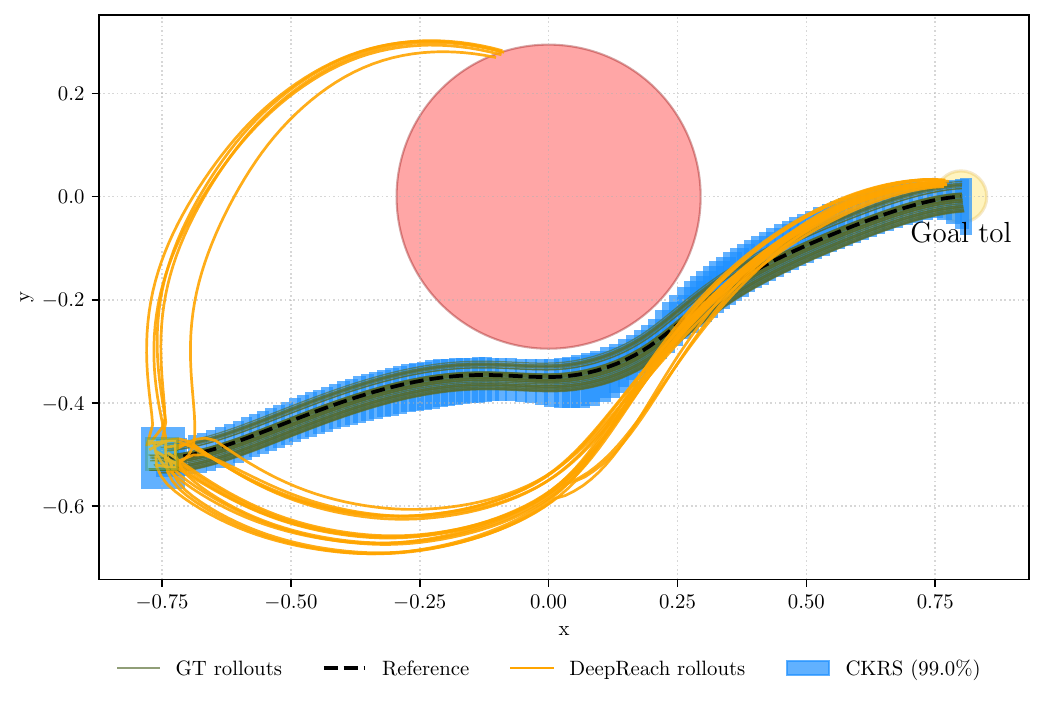}
  \caption{Closed-loop trajectories obtained from DeepReach-based controller (orange) can fail to enforce the reach-avoid constraint (reach yellow goal region while avoiding red unsafe set), while our controller (green rollouts) guarantees safety and containment within the blue CKRS with probability 0.99.}
  \label{fig:deepreach}
\end{figure}

\subsubsection{Baseline comparison with \cite{cp_baseline}.}\label{app:nnv}
Figure~\ref{fig:cpnnv} provides a state-by-state breakdown of the conservativeness gap between the reachable sets computed via the baseline CP-NNV \citep{cp_baseline} approach (yellow) and our proposed method (blue). While the tabular results in Section~\ref{sec:experiments} quantify the volume difference, Figure \ref{fig:cpnnv} illustrate where that difference originates visually. Overall, one of the key differences between our approach and \cite{cp_baseline} is in using a learned multi-step predictor instead of learned Koopman dynamics to compute reachable sets. We note that the CP-NNV method suffers from more conservative reachable sets, particularly in the angular velocity ($p, q, r$) and angle ($\phi, \theta, \psi$) dimensions. We attribute this to the difficulty of computing reachable sets accurately with a learned multi-step predictor over long horizons. Conversely, ScaRe-Kro maintains consistently tight bounds across all 12 dimensions, demonstrating that the Koopman linearization effectively mitigates the ``wrapping effect" that causes over-approximation in standard neural network verification.

\begin{figure}[ht]
    \centering
    \includegraphics[width=1\linewidth]{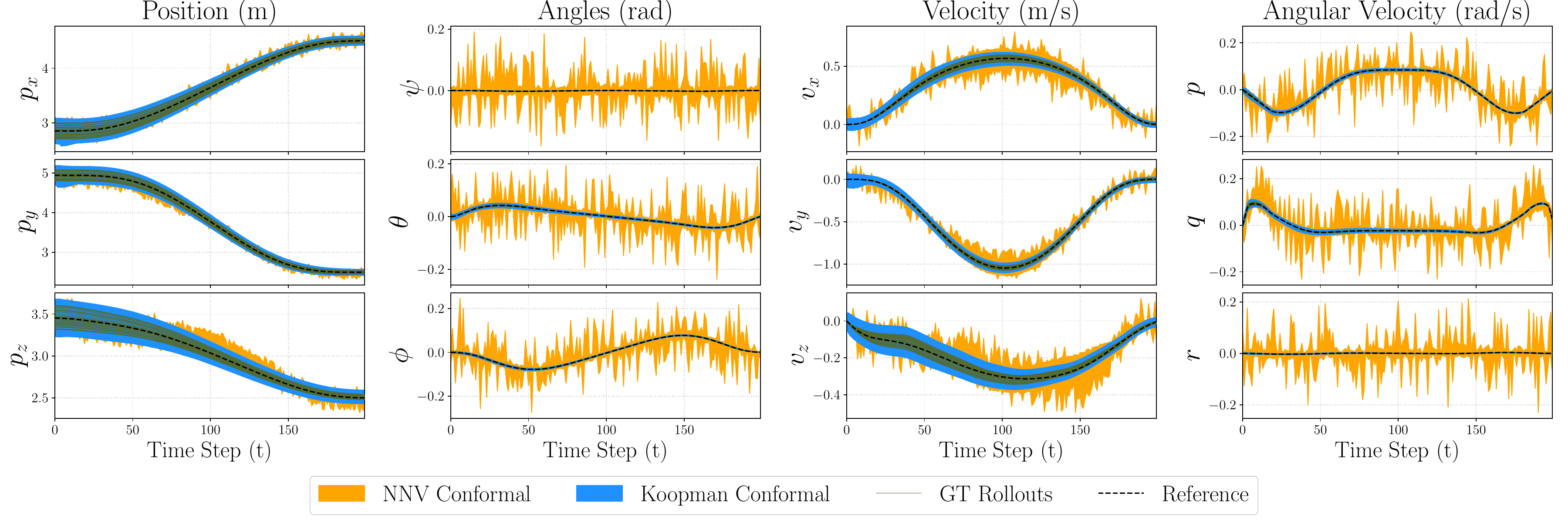}
    \caption{Baseline comparison between CP NNV approach and our method}
    \label{fig:cpnnv}
\end{figure}

\subsubsection{Additional results on unicycle system.}\label{app:unicycle}
Figure~\ref{fig:car} depicts the spatial evolution of the CKRS for the unicycle. Beyond the numerical coverage rates of Sec. \ref{sec:experiments}, this plot serves as a geometric sanity check, showing the reachable set's ability to contain the reference trajectory (black dashed line) and closed-loop trajectories (green) while navigating tight constraints. The blue tube is the CKRS that represents the region where the system is guaranteed to remain with with probability at least $1-\delta$, illustrating that even with the reachable set inflation induced by the conformalization step, the resulting set remains sufficiently compact to verify safety despite the closeness of the closed-loop rollouts to the red obstacle.

\begin{figure}[ht]
  \centering
  \includegraphics[width=1\linewidth]{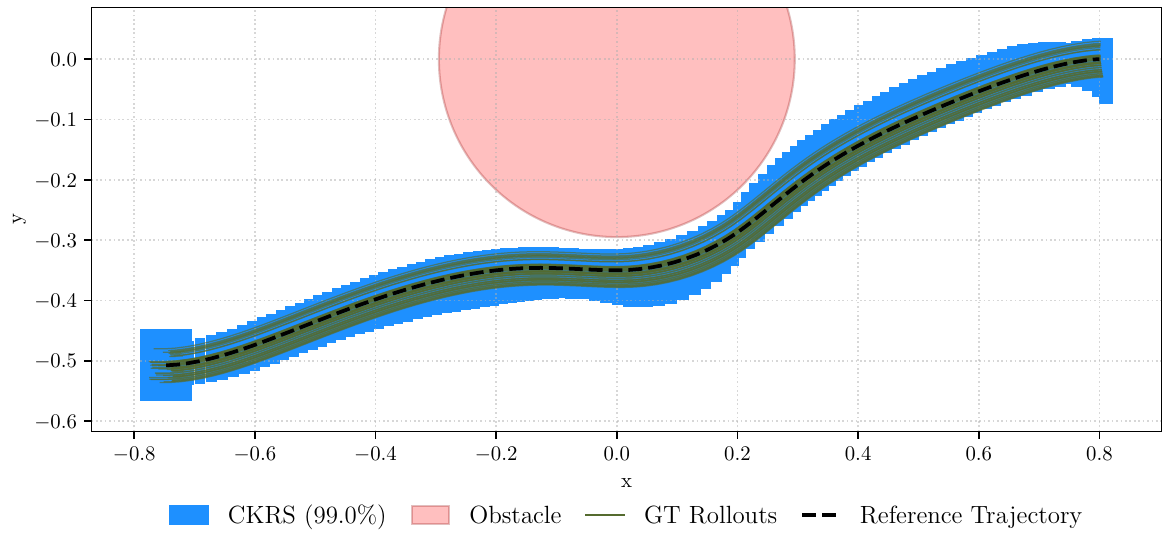}
  \caption{CKRS computed for 3D unicycle model, enabling safety verification despite the small distance of the closed-loop trajectories to the obstacle (red).}
  \label{fig:car}
\end{figure}

\subsubsection{Additional results on planar quadcopter system.}\label{app:quadcopter}
We present the full dimensional trajectory evolution for the 6D planar quadcopter in Figure~\ref{fig:2d_quad}. These plots isolate the performance of the Koopman tracking controller on an underactuated system. Notably, the reachable sets (light blue) and their conformalized counterparts (dark blue) accurately capture the behavior in all state dimensions, bounding the closed-loop trajectories. This confirms that the learned linear structure in the lifted Koopman space preserves the critical coupling effects of the original nonlinear dynamics, while the conformal bounds effectively absorb the residual error in the predicted KRS. We also note that the conformalization of the KRS in this example has a relatively small impact on the overall volume of the reachable set estimates, though we note that this is not always the case, e.g., the hopper system shown in Fig.~\ref{fig:hopper_exp}. Moreover, this experiment showcases the offline CP inflation bounds (described in Sec. \ref{sec:conformalize_KRS}) generalize across multiple reference trajectories, we compute the CKRS for four reference trajectories total (Fig. \ref{fig:2d_quad}a-d) and show that the resulting closed-loop trajectories remain within the computed set. This indicates that a single offline CP calibration can be reused to efficiently inflate reachable sets computed online.

\begin{figure}[ht]
    \centering
    \includegraphics[width=1\linewidth]{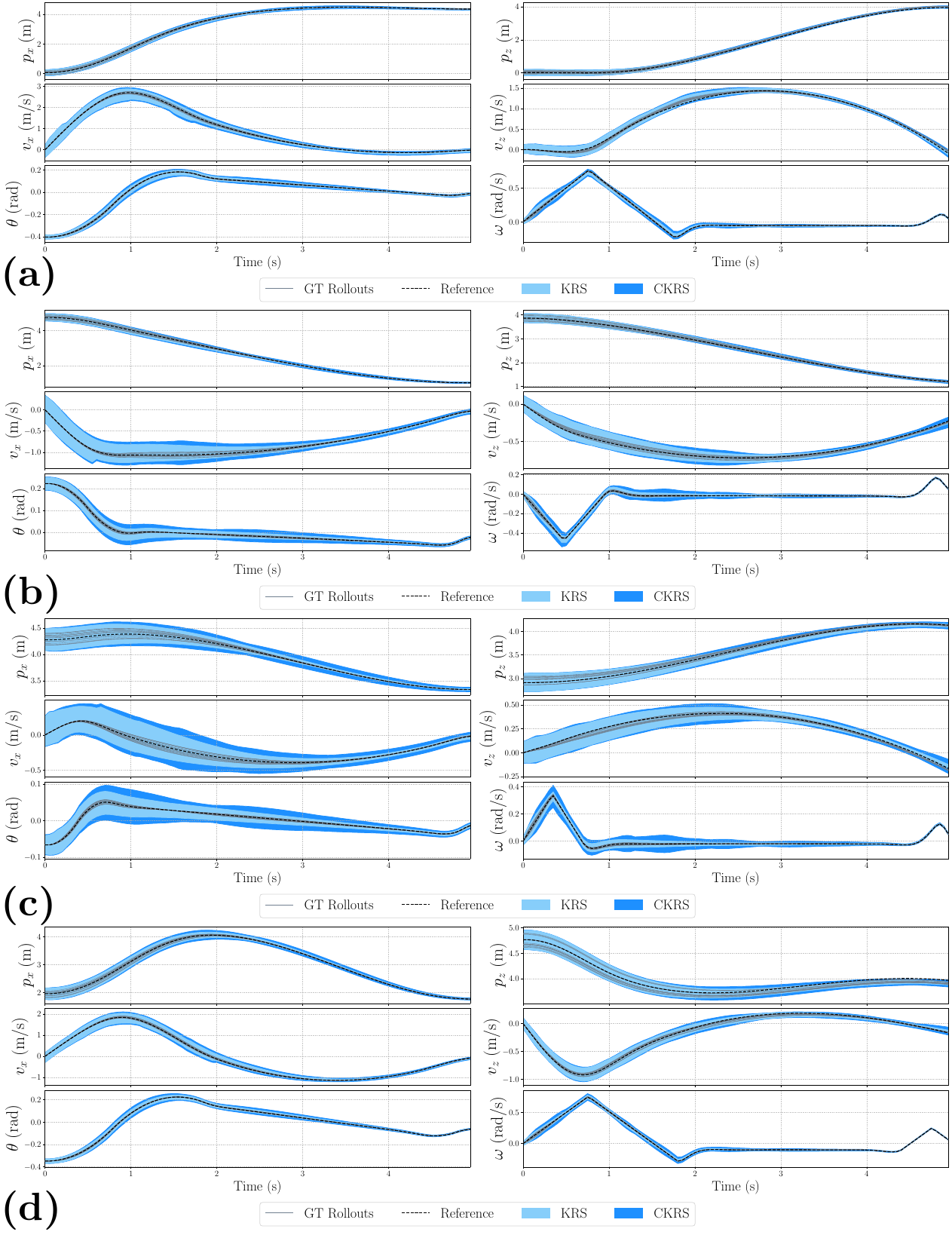}
    \caption{CKRS computed for four reference trajectories on a 2D quadcopter model, with offline-calibrated CP bounds.}
    \label{fig:2d_quad}
\end{figure}

\subsubsection{Additional results on 3D quadrotor system.}\label{app:3d_quad}

To verify that the offline CP inflation bounds (described in Sec. \ref{sec:conformalize_KRS}) generalize across multiple reference trajectories, we compute the CKRS for three alternative reference trajectories and show that the resulting closed-loop trajectories remain within the computed set (Fig. \ref{fig:3d_quad_alt}). This indicates that a single offline CP calibration can be reused to efficiently inflate reachable sets computed online. We also visualize two of these CKRSs together with the obstacle that the quadrotor is navigating around (Fig. \ref{fig:3d_quad_obstacle}), showing that since the RSOAs do not intersect with the obstacle, the closed-loop system is guaranteed to avoid collision when tracking either reference.
\begin{figure}[ht]
    \centering
    \includegraphics[width=1\linewidth]{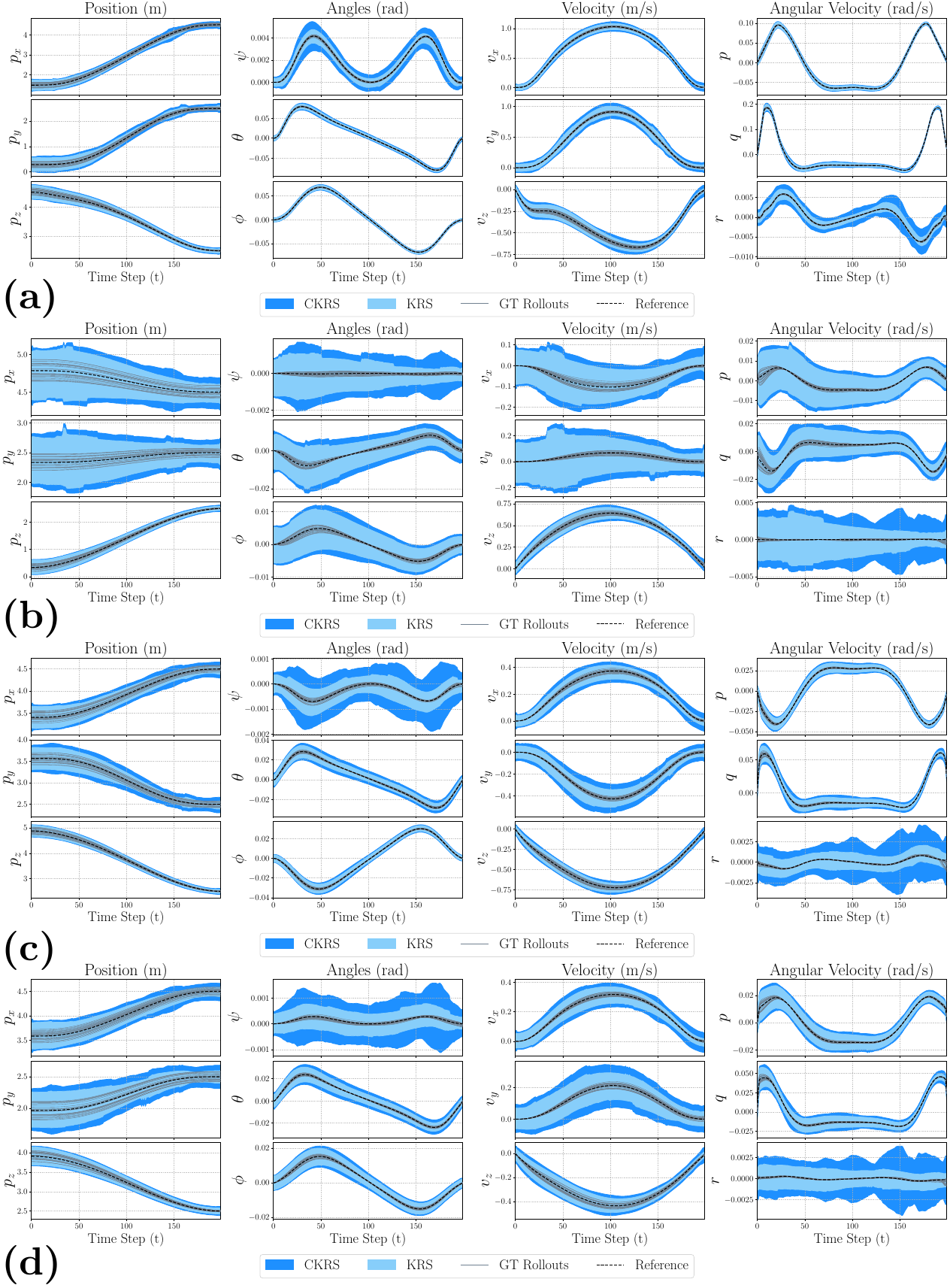}
    \caption{Alternative CKRSs computed for 3D quadcopter model for a different reference trajectory.}
    \label{fig:3d_quad_alt}
\end{figure}
\begin{figure}[ht]
    \centering
    \includegraphics[width=1\linewidth]{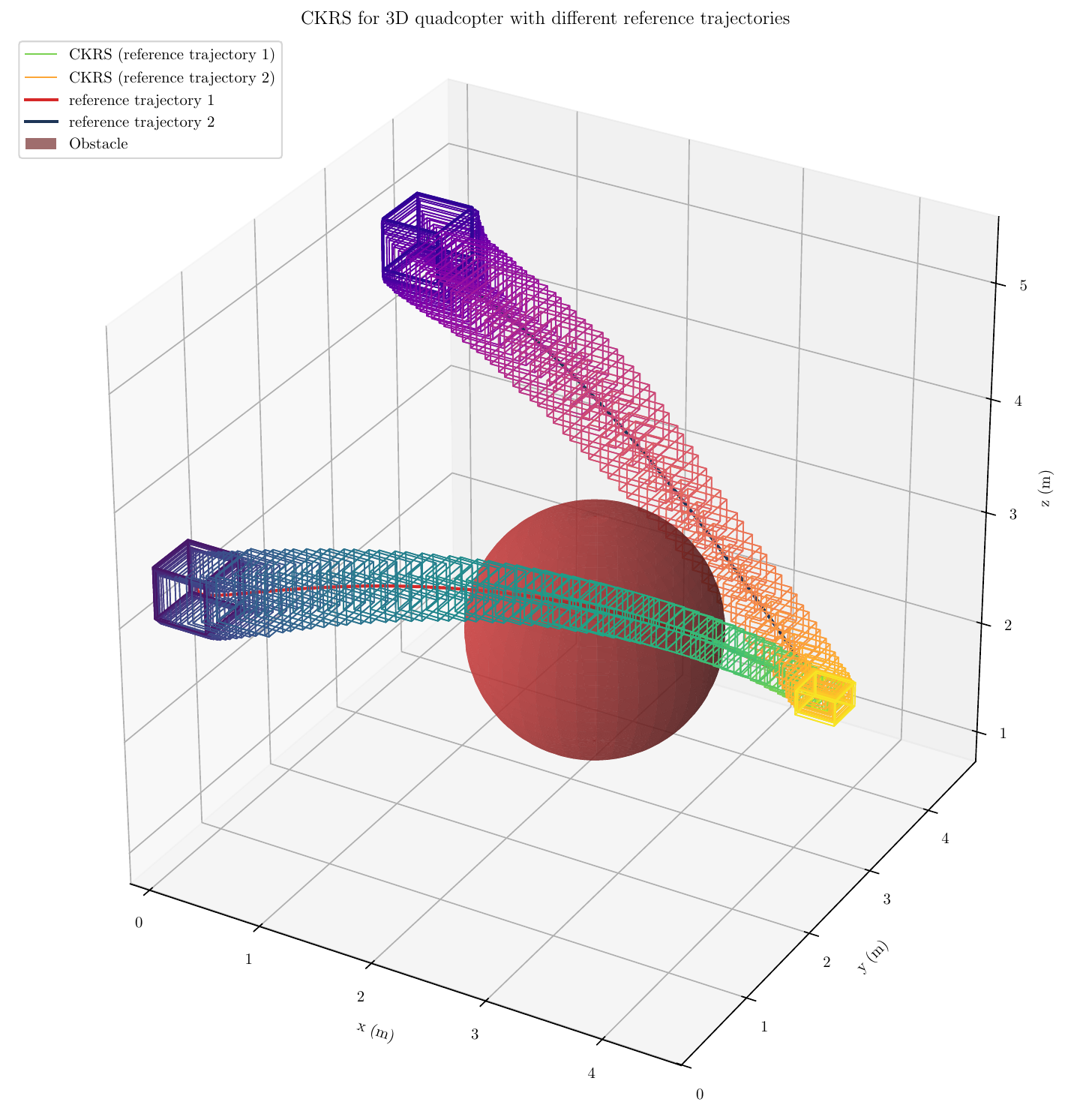}
    \caption{CKRSs computed for a 3D quadcopter model (position components), overlaid with obstacle.}
    \label{fig:3d_quad_obstacle}
\end{figure}

\subsubsection{Additional results on MuJoCo hopper system.}\label{app:hopper}
Figure~\ref{fig:hopper_exp} illustrates the full state-space evolution of the 11-dimensional MuJoCo Hopper \citep{todorov2012mujoco} system during a forward hopping maneuver. This system poses a unique challenge due to its hybrid dynamics, characterized by non-smooth transitions during ground impacts. The periodic stability of the gait is clearly visible in the ``height z" and ``torso angle" trajectories, where the ScaRe-Kro framework successfully captures the complex dynamics associated with the flight and stance phases. Notably, while the position states exhibit tight confinement around the reference trajectory (dashed black line), the velocity dimensions (e.g., ``Velocity Foot") display a necessary inflation of the conformal bounds. This expansion correctly accounts for the instantaneous jumps in state caused by contact forces and modeling uncertainties at impact, yet the ground truth rollouts remain strictly contained within the predicted tube, validating the robustness of the Koopman tracking controller under discontinuous contact dynamics.

\begin{figure}[ht]
    \centering
    \includegraphics[width=1\linewidth]{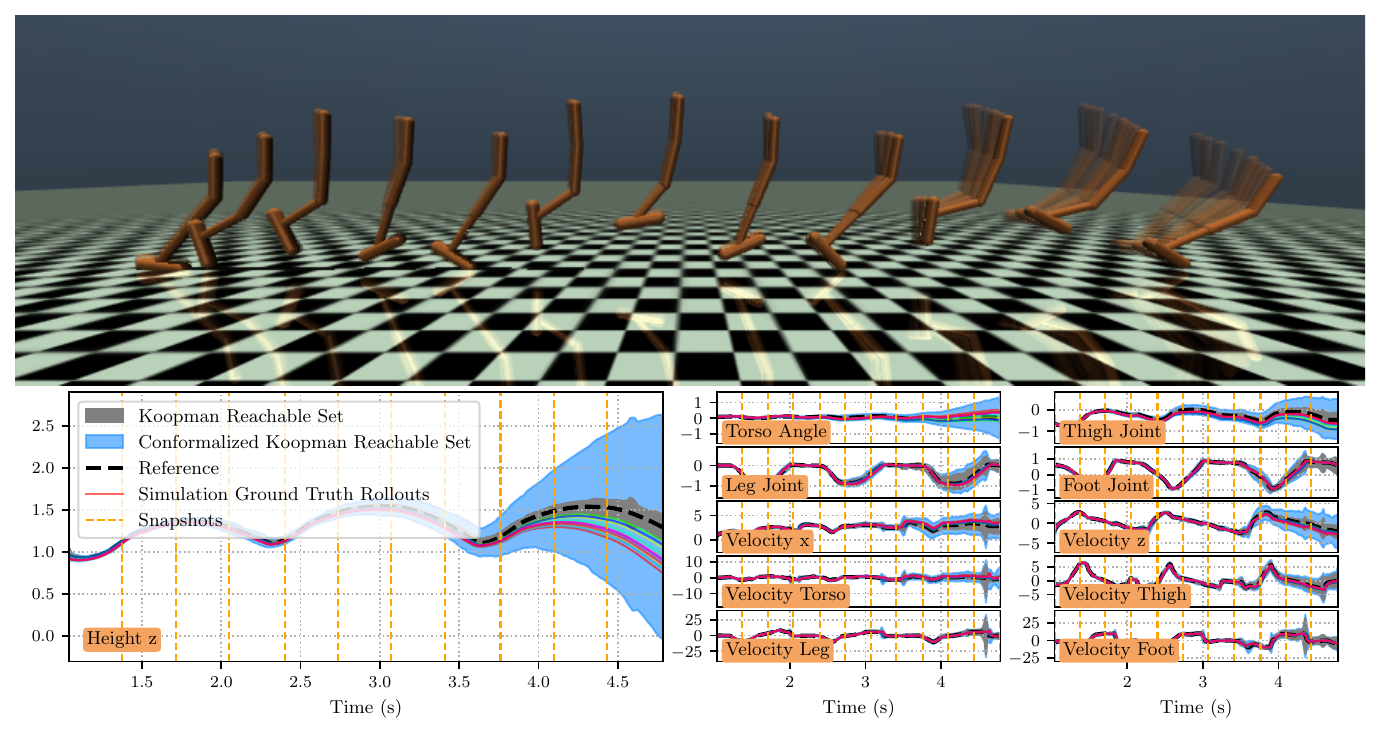}
    \caption{CKRS and KRS computed for 11-D MuJoCo Hopper. The thin multi-colored lines denote distinct closed-loop trajectory rollouts.}
    \label{fig:hopper_exp}
\end{figure}

\subsubsection{Additional results on MuJoCo swimmer system.}\label{app:swimmer}
Figure~\ref{fig:swimmer} presents the component-wise reachability analysis for the 28-dimensional MuJoCo Swimmer \cite{todorov2012mujoco}, representing the most high-dimensional validation of our framework. The plots display the highly coupled, undulatory gait required for locomotion, evidenced by the phase-shifted oscillations across the chain of link joints and velocities. Despite the significant increase in state dimensionality, the CKRS maintain consistent coverage of the closed-loop trajectories across all 28 dimensions. The uniform tightness of the bounds across the swimmer link segments confirms that the learned globally-linear dynamics structure effectively encodes the multi-body dynamics, enabling safe, long-horizon verification for complex, articulated robots while maintaining tractable computation.

\begin{figure}[ht]
    \centering
    \includegraphics[width=1\linewidth]{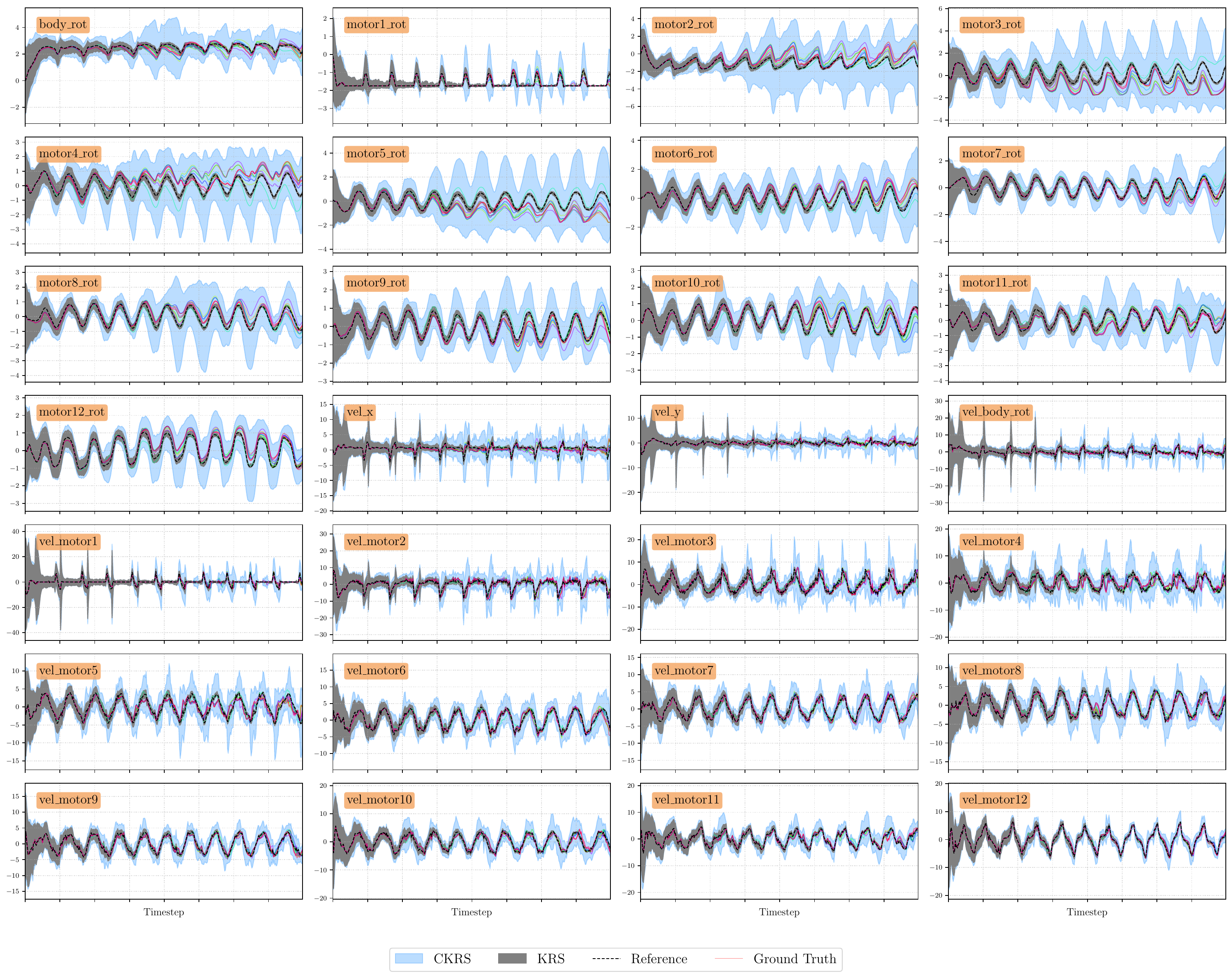}
    \caption{CKRS and KRS computed for 28-D MuJoCo Swimmer. The thin multi-colored lines denote distinct closed-loop trajectory rollouts.}
    \label{fig:swimmer}
\end{figure}
\end{document}